\RequirePackage{fix-cm}
\documentclass[smallextended]{svjour3}       % onecolumn (second format)
\smartqed  % flush right qed marks, e.g. at end of proof
\usepackage{graphicx}
\usepackage{mathrsfs}
\usepackage{amsfonts}
\usepackage{graphicx}
\usepackage{amssymb}
\usepackage{amsmath}
\usepackage{color}
\usepackage[misc]{ifsym}

\newtheorem{open problem}{Open Problem}
\journalname{%Designs, Codes and Cryptography
}
\begin{document}

\title{New results on vectorial dual-bent functions and partial difference sets%\thanks{Grants or other notes
}


\author{Jiaxin Wang  \and
        Fang-Wei Fu %etc.
}


\institute{%\Letter~
              Jiaxin Wang\at
              Chern Institute of Mathematics and LPMC, and Tianjin Key Laboratory of Network and Data Security Technology, Nankai University \\
              Tianjin, 300071, P. R. China\\
              \email{wjiaxin@mail.nankai.edu.cn}           %  \\
             \and
             Fang-Wei Fu\at
              Chern Institute of Mathematics and LPMC, and Tianjin Key Laboratory of Network and Data Security Technology, Nankai University \\
              Tianjin, 300071, P. R. China\\
              \email{fwfu@nankai.edu.cn}}

\date{Received: date / Accepted: date}

\maketitle

\begin{abstract}
  Bent functions $f: V_{n}\rightarrow \mathbb{F}_{p}$ with certain additional properties play an important role in constructing partial difference sets, where $V_{n}$ denotes an $n$-dimensional vector space over $\mathbb{F}_{p}$, $p$ is an odd prime. In \cite{Cesmelioglu1,Cesmelioglu2}, the so-called vectorial dual-bent functions are considered to construct partial difference sets. In \cite{Cesmelioglu1}, \c{C}e\c{s}melio\u{g}lu \emph{et al.} showed that for vectorial dual-bent functions $F: V_{n}\rightarrow V_{s}$ with certain additional properties, the preimage set of $0$ for $F$ forms a partial difference set. In \cite{Cesmelioglu2}, \c{C}e\c{s}melio\u{g}lu \emph{et al.} showed that for a class of Maiorana-McFarland vectorial dual-bent functions $F: V_{n}\rightarrow \mathbb{F}_{p^s}$, the preimage set of the squares (non-squares) in $\mathbb{F}_{p^s}^{*}$ for $F$ forms a partial difference set. In this paper, we further study vectorial dual-bent functions and partial difference sets. We prove that for vectorial dual-bent functions $F: V_{n}\rightarrow \mathbb{F}_{p^s}$ with certain additional properties, the preimage set of the squares (non-squares) in $\mathbb{F}_{p^s}^{*}$ for $F$ and the preimage set of any coset of some subgroup of $\mathbb{F}_{p^s}^{*}$ for $F$ form partial difference sets. Furthermore, explicit constructions of partial difference sets are yielded from some (non)-quadratic vectorial dual-bent functions. In this paper, we illustrate that almost all the results of using weakly regular $p$-ary bent functions to construct partial difference sets are special cases of our results. In \cite{Cesmelioglu1}, \c{C}e\c{s}melio\u{g}lu \emph{et al.} showed that for a weakly regular $p$-ary bent function $f$ with $f(0)=0$, if $f$ is a so-called $l$-form with $gcd(l-1, p-1)=1$ for some integer $1\leq l\leq p-1$, then $f$ is vectorial dual-bent. We prove that the converse also holds, which answers one open problem proposed in \cite{Cesmelioglu2}.
\end{abstract}

\section{Introduction}
\label{intro}
Throughout this paper, let $p$ be an odd prime, $\zeta_{p}=e^{\frac{2 \pi \sqrt{-1}}{p}}$ be the complex primitive $p$-th root of unity, $\mathbb{F}_{p}^{n}$ be the vector space of the $n$-tuples over $\mathbb{F}_{p}$, $\mathbb{F}_{p^n}$ be the finite field with $p^n$ elements, $V_{n}$ be an $n$-dimensional vector space over $\mathbb{F}_{p}$, and $\langle \cdot \rangle_{n}$ denote a (non-degenerate) inner product of $V_{n}$, where $n$ is a positive integer. In this paper, when $V_{n}=\mathbb{F}_{p}^{n}$, let $\langle a, b\rangle_{n}=a \cdot b=\sum_{i=1}^{n}a_{i}b_{i}$, where $a=(a_{1}, \dots, a_{n}), b=(b_{1}, \dots, b_{n})\in \mathbb{F}_{p}^{n}$; when $V_{n}=\mathbb{F}_{p^n}$, let $\langle a, b\rangle_{n}=Tr_{1}^{n}(ab)$, where $a, b \in \mathbb{F}_{p^n}$, $Tr_{k}^{n}(\cdot)$ denotes the trace function from $\mathbb{F}_{p^n}$ to $\mathbb{F}_{p^k}$, $k$ is a divisor of $n$; when $V_{n}=V_{n_{1}}\times \dots \times V_{n_{m}} (n=\sum_{i=1}^{m}n_{i})$, let $\langle a, b\rangle_{n}=\sum_{i=1}^{m}\langle a_{i}, b_{i}\rangle_{n_{i}}$, where $a=(a_{1}, \dots, a_{m}), b=(b_{1}, \dots, b_{m})\in V_{n}$.

A function $f$ from $V_{n}$ to $V_{s}$ is called a vectorial $p$-ary function, or simply $p$-ary function when $s=1$. The Walsh transform of a $p$-ary function $f: V_{n}\rightarrow \mathbb{F}_{p}$ is defined as
\begin{equation}\label{1}
  W_{f}(a)=\sum_{x \in V_{n}}\zeta_{p}^{f(x)-\langle a, x\rangle_{n}}, a \in V_{n}.
\end{equation}

A $p$-ary function $f: V_{n}\rightarrow \mathbb{F}_{p}$ is called bent if $|W_{f}(a)|=p^{\frac{n}{2}}$ for any $a \in V_{n}$. A vectorial $p$-ary function $F: V_{n}\rightarrow V_{s}$ is called vectorial bent if all component functions $F_{c}, c \in V_{s}^{*}=V_{s}\backslash \{0\}$ defined as $F_{c}(x)=\langle c, F(x)\rangle_{s}$ are bent. Note that if $f: V_{n}\rightarrow \mathbb{F}_{p}$ is bent, then so are $cf, c \in \mathbb{F}_{p}^{*}$, that is, bent functions are vectorial bent. For a vectorial $p$-ary function $F: V_{n}\rightarrow V_{s}$, the vectorial bentness of $F$ is independent of the inner products of $V_{n}$ and $V_{s}$. The Walsh transform of a $p$-ary bent function $f: V_{n}\rightarrow \mathbb{F}_{p}$ satisfies that for any $a \in V_{n}$,
\begin{equation}\label{2}
  W_{f}(a)=\left\{\begin{split}
                     \pm p^{\frac{n}{2}}\zeta_{p}^{f^{*}(a)} & \ \ \text{if} \ p \equiv 1 \ (mod \ 4) \ \text{or} \ n \ \text{is even},\\
                     \pm \sqrt{-1} p^{\frac{n}{2}} \zeta_{p}^{f^{*}(a)} & \ \ \text{if} \ p \equiv 3 \ (mod \ 4) \ \text{and} \ n \ \text{is odd},
                  \end{split}\right.
\end{equation}
where $f^{*}$ is a function from $V_{n}$ to $\mathbb{F}_{p}$, called the dual of $f$. A $p$-ary bent function $f: V_{n}\rightarrow \mathbb{F}_{p}$ is said to be weakly regular if $W_{f}(a)=\varepsilon_{f}p^{\frac{n}{2}}\zeta_{p}^{f^{*}(a)}$, where $\varepsilon_{f}$ is a constant independent of $a$, otherwise $f$ is called non-weakly regular. In particular, if $\varepsilon_{f}=1$, $f$ is said to be regular. The (non)-weakly regularity of $f$ is independent of the inner product of $V_{n}$ and if $f$ is weakly regular, $\varepsilon_{f}$ is independent of the inner product of $V_{n}$.

In \cite{Cesmelioglu3}, \c{C}e\c{s}melio\u{g}lu \emph{et al.} introduced vectorial dual-bent functions. A vectorial $p$-ary bent function $F: V_{n}\rightarrow V_{s}$ is said to be vectorial dual-bent if there exists a vectorial bent function $G: V_{n}\rightarrow V_{s}$ such that $(F_{c})^{*}=G_{\sigma(c)}$ for any $c \in V_{s}^{*}$, where $(F_{c})^{*}$ is the dual of the component function $\langle c, F(x)\rangle_{s}$ and $\sigma$ is some permutation over $V_{s}^{*}$. The vectorial bent function $G$ is called a vectorial dual of $F$ and denoted by $F^{*}$. It is known in \cite{Cesmelioglu3} that the property of being vectorial dual-bent is independent of the inner products of $V_{n}$ and $V_{s}$. A $p$-ary function $f: V_{n}\rightarrow \mathbb{F}_{p}$ is called an $l$-form if $f(ax)=a^{l}f(x)$ for any $a \in \mathbb{F}_{p}^{*}$ and $x \in V_{n}$, where $1\leq l\leq p-1$ is an integer. In \cite{Cesmelioglu1}, \c{C}e\c{s}melio\u{g}lu \emph{et al.} showed that for a weakly regular $p$-ary bent function $f: V_{n}\rightarrow \mathbb{F}_{p}$ with $f(0)=0$, if $f$ is an $l$-form with $gcd(l-1, p-1)=1$ for some integer $1\leq l\leq p-1$, then $f$ (seen as a vectorial bent function from $V_{n}$ to $V_{1}$) is vectorial dual-bent. In this paper, we will prove that the converse also holds, which answers one open problem proposed in \cite{Cesmelioglu2}.

Let $G$ be a finite abelian group of order $v$. A subset $D$ of $G$ with $k$ elements is called a $(v,k,\lambda,\mu)$ partial difference set of $G$, if every nonzero element of $D$ can be represented as $d_{1}-d_{2}, d_{1}, d_{2} \in D$, in exactly $\lambda$ ways, and every nonzero element of $G \ \backslash \ D$ can be represented as $d_{1}-d_{2}, d_{1}, d_{2} \in D$, in exactly $\mu$ ways. When $\lambda=\mu$, $D$ is called a $(v, k, \lambda)$ difference set. A partial difference set $D$ is called regular if $-D=D$ and $0 \notin D$. If $D$ is a partial difference set of $G$ with $-D=D$, then so are $D \cup \{0\}, D\ \backslash \ \{0\}, G \ \backslash \ D$. In this paper, we consider regular partial difference sets. There is an equivalence way to describe regular partial difference sets in terms of strongly regular graphs \cite{Ma}.

Bent functions play an important role in constructing partial difference sets (or equivalently, strongly regular graphs). It is known that the support of a Boolean bent function $f$, respectively the preimage set of $0$ for $f$, is a Hadamard difference set (see \cite{Hou}). In 2010, Tan \emph{et al.} \cite{Tan} for the first time used weakly regular ternary bent functions $f: V_{2n}\rightarrow \mathbb{F}_{3}$ of $2$-form with $f(0)=0$ to construct partial difference sets. Then their works were extended to weakly regular $p$-ary bent functions for any odd prime $p$ \cite{Chee,Feng,Hyun}. In Section 2, we will review the works in these papers. In \cite{Cesmelioglu1} and \cite{Cesmelioglu2}, vectorial dual-bent functions are considered to construct partial difference sets. In \cite{Cesmelioglu1}, \c{C}e\c{s}melio\u{g}lu \emph{et al.} showed that for vectorial dual-bent functions $F$ with certain additional properties, the preimage set of $0$ for $F$ forms a partial difference set. In \cite{Cesmelioglu2}, \c{C}e\c{s}melio\u{g}lu \emph{et al.} showed that for a class of Maiorana-McFarland vectorial dual-bent functions $F: V_{n}\rightarrow \mathbb{F}_{p^s}$, the preimage set of the squares (non-squares) in $\mathbb{F}_{p^s}^{*}$ for $F$ forms a partial difference set. In fact, as nonsingular quadratic forms $F: \mathbb{F}_{q}^{n}\rightarrow \mathbb{F}_{q}$ ($q$ is an odd prime power) can be proved to be vectorial dual-bent, the works of using vectorial dual-bent functions to construct partial difference sets can be traced to \cite{Calderbank,Feng,Ma} in which partial differences sets were yielded from nonsingular quadratic forms. In this paper, we further study vectorial dual-bent functions and partial difference sets. We provide a new explicit construction of vectorial dual-bent functions. We prove that for vectorial dual-bent functions $F: V_{n}\rightarrow \mathbb{F}_{p^s}$ with certain additional properties, if the permutation $\sigma$ of $\mathbb{F}_{p^s}^{*}$ given by $(F_{c})^{*}=(F^{*})_{\sigma(c)}, c \in \mathbb{F}_{p^s}^{*}$ (where $F^{*}$ is a vectorial dual of $F$) is an identity map, then the preimage set of any subset of $\mathbb{F}_{p^s}$ for $F$ forms a partial difference set. For vectorial dual-bent functions $F: V_{n}\rightarrow \mathbb{F}_{p^s}$ with certain additional properties, if the permutation $\sigma$ of $\mathbb{F}_{p^s}^{*}$ given by $(F_{c})^{*}=(F^{*})_{\sigma(c)}, c \in \mathbb{F}_{p^s}^{*}$ (where $F^{*}$ is a vectorial dual of $F$) is not an identity map but it satisfies some conditions, we prove that the preimage set of the squares (non-squares) in $\mathbb{F}_{p^s}^{*}$ for $F$ and the preimage set of any coset of some subgroup of $\mathbb{F}_{p^s}^{*}$ for $F$ form partial difference sets. Furthermore, explicit constructions of partial difference sets are yielded from some (non)-quadratic vectorial dual-bent functions. In this paper, we illustrate that almost all the results of using weakly regular $p$-ary bent functions to construct partial difference sets are special cases of our results. In \cite{Cesmelioglu1}, \c{C}e\c{s}melio\u{g}lu \emph{et al.} showed that for a weakly regular $p$-ary bent function $f: V_{n}\rightarrow \mathbb{F}_{p}$ with $f(0)=0$, if $f$ is an $l$-form with $gcd(l-1, p-1)=1$ for some integer $1\leq l \leq p-1$, then $f$ (seen as a vectorial bent function from $V_{n}$ to $V_{1}$) is vectorial dual-bent. We prove that the converse also holds, which answers one open problem proposed in \cite{Cesmelioglu2}.

The rest of the paper is organized as follows. In Section 2, we review the related results on vectorial dual-bent functions and partial difference sets. In Section 3, we provide a new explicit construction of vectorial dual-bent functions. In Section 4, we provide new results on vectorial dual-bent functions and partial difference sets, and we illustrate that almost all the results of using weakly regular $p$-ary bent functions to construct partial difference sets are special cases of our results. In Section 5, we make a conclusion.
\section{Preliminaries}
\label{sec:2}
In this section, we review the related results on vectorial dual-bent functions and partial difference sets.

Let $L$ be a linear permutation of $V_{n}$. If for two vectorial functions $F: V_{n}\rightarrow V_{s}$ and $G: V_{n}\rightarrow V_{s}$ we have $G(x)=F(L(x))$, then $F$ and $G$ are called linear equivalent. By the results in \cite{Cesmelioglu3}, if $F: V_{n}\rightarrow V_{s}$ is a vectorial dual-bent function, then $F(L(x))$ is also a vectorial dual-bent function, that is, the property of being vectorial dual-bent is an invariant under linear equivalence. In \cite{Cesmelioglu2}, \c{C}e\c{s}melio\u{g}lu \emph{et al.} showed that if $D_{A}^{F}=\{x \in V_{n}: F(x)\in A\}$ is a $(v, k, \lambda, \mu)$ partial difference set for some vectorial function $F: V_{n}\rightarrow V_{s}$ and subset $A$ of $V_{s}$, then $D_{A}^{G}=\{x \in V_{n}: G(x) \in A\}$ is also a $(v, k, \lambda, \mu)$ partial difference set, where $G(x)=F(L(x))$. Therefore, when we consider using vectorial dual-bent functions to construct partial difference sets, we only need to consider vectorial dual-bent functions in the sense of linear equivalence. Up to now, there are only a few known vectorial dual-bent functions (in the linear equivalence) \cite{Cesmelioglu3}. In Section 3, we will provide a new explicit construction of (non)-quadratic vectorial dual-bent functions. Below we list the vectorial dual-bent functions given in \cite{Cesmelioglu3}.

(i) Let $s$ be a divisor of $m$. Let $F: \mathbb{F}_{p^m} \times \mathbb{F}_{p^m}\rightarrow \mathbb{F}_{p^s}$ be defined as
\begin{equation}\label{3}
  F(x, y)=Tr_{s}^{m}(axy^{e}),
\end{equation}
where $a \in \mathbb{F}_{p^m}^{*}, gcd(e, p^m-1)=1$. Then $F$ is vectorial dual-bent for which for any $c \in \mathbb{F}_{p^s}^{*}$, $\varepsilon_{F_{c}}=1$, $(F_{c})^{*}=(F^{*})_{\sigma(c)}$, where $F^{*}(x, y)=Tr_{s}^{m}(-a^{-u}x^{u}y)$, $\sigma(c)=c^{-u}$ with $eu\equiv 1 \ mod \ (p^m-1)$.

(ii) Let $s$ be a divisor of $m$. Let $F: \mathbb{F}_{p^m} \times \mathbb{F}_{p^m}\rightarrow \mathbb{F}_{p^s}$ be defined as
\begin{equation}\label{4}
  F(x, y)=Tr_{s}^{m}(axL(y)),
\end{equation}
where $L(x)=\sum a_{i}x^{q^{i}} \ (q=p^s)$ is a $q$-polynomial over $\mathbb{F}_{p^m}$ inducing a permutation of $\mathbb{F}_{p^m}$ and $a \in \mathbb{F}_{p^m}^{*}$. Then $F$ is vectorial dual-bent for which for any $c \in \mathbb{F}_{p^s}^{*}$, $\varepsilon_{F_{c}}=1$, $(F_{c})^{*}=(F^{*})_{\sigma(c)}$, where $F^{*}(x, y)=Tr_{s}^{m}(-L^{-1}(a^{-1}x)y)$, $\sigma(c)=c^{-1}$.

(iii) In this paper, let $\eta_{n}$ denote the multiplicative quadratic character of $\mathbb{F}_{p^n}$, that is $\eta_{n}(a)=1$ if $a\in \mathbb{F}_{p^n}^{*}$ is a square and $\eta_{n}(a)=-1$ if $a\in \mathbb{F}_{p^n}^{*}$ is a non-square. Let $s$ be a divisor of $n$ and $F: \mathbb{F}_{p^n}\rightarrow \mathbb{F}_{p^s}$ be defined as
\begin{equation}\label{5}
  F(x)=Tr_{s}^{n}(ax^{2}),
\end{equation}
where $a \in \mathbb{F}_{p^n}^{*}$. Then $F$ is vectorial dual-bent for which for any $c \in \mathbb{F}_{p^s}^{*}$, $\varepsilon_{F_{c}}=(-1)^{n-1}\epsilon^{n}\eta_{n}(ac)$, $(F_{c})^{*}=(F^{*})_{\sigma(c)}$, where $F^{*}(x)=Tr_{s}^{n}(-x^{2}/(4a)), \sigma(c)=c^{-1}$, $\epsilon=1$ if $p  \equiv 1 \ (mod \ 4)$ and $\epsilon=\sqrt{-1}$ if $p \equiv 3 \ (mod \ 4)$.

(iv) Consider a $2m$-dimensional vector space $\mathbb{F}=\mathbb{F}_{p^m} \times \mathbb{F}_{p^m}$ (or $\mathbb{F}=\mathbb{F}_{p^{2m}}$) over $\mathbb{F}_{p}$. A partial spread of $\mathbb{F}$ is a collection of subspaces $U_{0}, U_{1}, \dots, U_{r}$ of $\mathbb{F}$ of dimension $m$ such that $U_{i} \cap U_{j}=\{0\}$ if $i\neq j$. If $r=p^m$, then the collection of subspaces $U_{0}, U_{1}, \dots, U_{p^m}$ is called a spread. For a subspace $U$ of $\mathbb{F}$, let $U^{\perp}=\{(z_{1}, z_{2})\in \mathbb{F}: Tr_{1}^{m}(z_{1}x_{1}+z_{2}x_{2})=0 \ \text{for all} \ (x_{1}, x_{2})\in U\}$ if $\mathbb{F}=\mathbb{F}_{p^m} \times \mathbb{F}_{p^m}$, and
$U^{\perp}=\{z\in \mathbb{F}: Tr_{1}^{2m}(zx)=0 \ \text{for all} \ x \in U\}$ if $\mathbb{F}=\mathbb{F}_{p^{2m}}$. For a positive integer $s\leq m$, let $g$ be a balanced function from $\{1, 2, \dots, p^m\}$ to $V_{s}$, that is, for any $a \in V_{s}$, $|\{1\leq i\leq p^m: g(i)=a\}|=p^{m-s}$. Suppose $g(i)=\gamma_{i}, 1\leq i \leq p^m$. Let $\gamma_{0}$ be an arbitrary element of $V_{s}$. Define $F: \mathbb{F}\rightarrow V_{s}$ as
\begin{equation*}
  F(z)=\gamma_{i} \ \text{if} \ z \in U_{i}, z \neq 0, 1\leq i \leq p^m, \ \text{and} \ F(z)=\gamma_{0} \ \text{if} \ z \in U_{0},
\end{equation*}
where $U_{0}, U_{1}, \dots, U_{p^m}$ form a spread of $\mathbb{F}$. Then $F$ is called a vectorial partial spread bent function and it is vectorial dual-bent for which for any $c \in V_{s}^{*}$, $\varepsilon_{F_{c}}=1$, $(F_{c})^{*}=(F^{*})_{c}$, where $F^{*}(z)=\gamma_{i}$ if $z \in U_{i}^{\perp}, z \neq 0, 1\leq i \leq p^m$, and $F^{*}(z)=\gamma_{0}$ if $z \in U_{0}^{\perp}$.

Note that the vectorial dual-bent functions defined by (4) and (5) are quadratic. For a general nonsingular quadratic form $F: \mathbb{F}_{p^s}^{m}\rightarrow \mathbb{F}_{p^s}$ defined as $F(x_{1}, \dots, x_{m})=\sum_{i, j=1}^{m}a_{i, j}x_{i}x_{j}$, where $a_{i, j}=a_{j, i}\in \mathbb{F}_{p^s}$, by Theorem 6.21 of \cite{Lidl}, $F$ is linear equivalent to a nonsingular diagonal quadratic form
\begin{equation}\label{6}
  G(x_{1}, \dots, x_{m})=a_{1}x_{1}^{2}+\dots+a_{m}x_{m}^{2},
\end{equation}
where $a_{i} \in \mathbb{F}_{p^s}^{*}, 1\leq i \leq m$. Then $G$ is vectorial dual-bent for which for any $c \in \mathbb{F}_{p^s}^{*}$, $\varepsilon_{G_{c}}=(-1)^{(s-1)m}\epsilon^{sm}\eta_{s}(c^{m}a_{1}\cdots a_{m})$, $(G_{c})^{*}=(G^{*})_{\sigma(c)}$, where $G^{*}(x_{1}, \dots, x_{m})=-x_{1}^{2}/(4a_{1})-\dots-x_{m}^{2}/(4a_{m})$, $\sigma(c)=c^{-1}$, $\epsilon=1$ if $p  \equiv 1 \ (mod \ 4)$ and $\epsilon=\sqrt{-1}$ if $p \equiv 3 \ (mod \ 4)$. Since $F$ and $G$ are linear equivalent, $F$ is also vectorial dual-bent.

For any $F: V_{n}\rightarrow \mathbb{F}_{p^s}$, define
\begin{equation}\label{7}
  \begin{split}
   &D_{0}^{F}=\{x \in V_{n}^{*}: F(x)=0\},\\
   &D_{S}^{F}=\{x \in V_{n}^{*}: F(x) \ \text{is  a  square in} \ \mathbb{F}_{p^s}^{*}\},\\
   &D_{N}^{F}=\{x \in V_{n}^{*}: F(x) \ \text{is  a  non\raisebox{0mm}{-}square in} \ \mathbb{F}_{p^s}^{*}\},\\
   &D_{\beta H_{l}}^{F}=\{x \in V_{n}^{*}: F(x) \in \beta H_{l}\}, \beta \in \mathbb{F}_{p^s}^{*},
  \end{split}
\end{equation}
where $H_{l}$ is a subgroup of $\mathbb{F}_{p^s}^{*}$ defined as $H_{l}=\{x^{l}, x \in \mathbb{F}_{p^s}^{*}\}$ for some positive integer $l$ and $\beta H_{l}$ is the coset of $H_{l}$ containing $\beta$. Note that if $l=2$, then $D_{\beta H_{l}}^{F}=D_{S}^{F}$ if $\beta \in \mathbb{F}_{p^s}^{*}$ is a square, and $D_{\beta H_{l}}^{F}=D_{N}^{F}$ if $\beta \in \mathbb{F}_{p^s}^{*}$ is a non-square.

The following two propositions give the known results of using weakly regular $p$-ary bent functions and vectorial dual-bent functions to construct partial difference sets respectively.
\begin{proposition}{\emph{\cite{Chee,Feng,Hyun,Tan}}}\label{Proposition1}
Let $n$ be a positive even integer. Let $f: V_{n}\rightarrow \mathbb{F}_{p}$ be a weakly regular $p$-ary bent function of $l$-form with $f(0)=0$ and $gcd(l-1, p-1)=1$.

\emph{(i)} If $gcd(l, p-1)\neq 2$, then $f$ must be regular;

\emph{(ii)} Then $D_{0}^{f}$, $D_{S}^{f}$, $D_{N}^{f}$ are partial difference sets and any union of sets from $\{D_{0}^{f}, D_{S}^{f}, D_{N}^{f}\}$ is a partial difference set;

\emph{(iii)} If $l\neq p-1$ when $p>3$, then $D_{\beta H_{l}}^{f}, \beta \in \mathbb{F}_{p}^{*}$ are partial difference sets and any union of sets from $\{D_{0}^{f}, D_{\beta H_{l}}^{f}, \beta \in \mathbb{F}_{p}^{*}\}$ is a partial difference set.
\end{proposition}

\begin{proposition}{\emph{\cite{Calderbank,Cesmelioglu1,Cesmelioglu2,Feng,Ma}}}\label{Proposition2}
Let $n$ be a positive even integer. Let $F: V_{n}\rightarrow V_{s}$ be a vectorial dual-bent function which satisfies that $F(0)=0, F(-x)=F(x)$ and all component functions are regular or weakly regular but not regular.

\emph{(i)} Then $D_{0}^{F}$ is a partial difference set;

\emph{(ii)} If $F=(f_{1}, \dots, f_{s}): V_{n} \rightarrow \mathbb{F}_{p}^{s}$ and suppose $F_{1}=(f_{1}, \dots, f_{k}): V_{n} \rightarrow \mathbb{F}_{p}^{k}$ is also vectorial dual-bent, then $D_{0}^{F_{1}} \ \backslash \ D_{0}^{F}$ is a partial difference set;

\emph{(iii)} If $F: \mathbb{F}_{p^m} \times \mathbb{F}_{p^m}\rightarrow \mathbb{F}_{p^s}$ is the Maiorana-McFarland vectorial bent function $F(x, y)=Tr_{s}^{m}(xL(y))$, where $s$ is a divisor of $m$ and $L$ is a linear permutation over $\mathbb{F}_{p^m}$, then $D_{S}^{F}, D_{N}^{F}$ are partial difference sets;

\emph{(iv)} If $F: \mathbb{F}_{p^s}^{m} \rightarrow \mathbb{F}_{p^s}$ is a nonsingular quadratic form for which $m$ is even, then $D_{S}^{F}$ is a partial difference set;

\emph{(v)} If $F: \mathbb{F}_{p^s}^{m} \rightarrow \mathbb{F}_{p^s}$ is a nonsingular quadratic form for which $m$ is even, $s=2jr$ for some positive integers $j, r$, and there is $t\geq 2$ such that $t \mid (p^{j}+1)$, where $j$ is the smallest such positive integer, then $D_{\beta H_{t}}^{F}, \beta \in \mathbb{F}_{p^s}^{*}$ are partial difference sets.
\end{proposition}

There is an important tool to characterize partial difference sets in terms of characters.

\begin{lemma}{\emph{\cite{Ma}}}\label{Lemma1}
Let $G$ be an abelian group of order $v$. Suppose that $D$ is a subset of $G$ with $k$ elements which satisfies $-D=D$ and $0 \notin D$. Then $D$ is a $(v, k, \lambda, \mu)$ partial difference set if and only if for each non-principal character $\chi$ of $G$,
\begin{equation*}
  \chi(D)=\frac{\beta\pm \sqrt{\Delta}}{2},
\end{equation*}
where $\chi(D):=\sum_{x \in D}\chi(x)$, $\beta=\lambda-\mu, \gamma=k-\mu, \Delta=\beta^{2}+4\gamma$.
\end{lemma}

\section{A new explicit construction of vectorial dual-bent functions}
\label{sec:3}
In this section, we provide a new explicit construction of vectorial dual-bent functions.

\begin{theorem}\label{Theorem1}
Let $s, m, n$ be positive integers with $s \mid m$ and $s \mid n$. Define $F_{i} \ (i \in \mathbb{F}_{p^s}): \mathbb{F}_{p^n}\rightarrow \mathbb{F}_{p^s}$ as
$F_{0}(x)=Tr_{s}^{n}(\alpha_{1}x^{2})$, $F_{i}(x)=Tr_{s}^{n}(\alpha_{2}x^{2})$ if $i$ is a square in $\mathbb{F}_{p^s}^{*}$, $F_{i}(x)=Tr_{s}^{n}(\alpha_{3}x^{2})$ if $i$ is a non-square in $\mathbb{F}_{p^s}^{*}$, where $\alpha_{j} \in \mathbb{F}_{p^n}^{*}, 1\leq j\leq 3$. Define $G: \mathbb{F}_{p^m} \times \mathbb{F}_{p^m}\rightarrow \mathbb{F}_{p^s}$ as $G(y_{1}, y_{2})=Tr_{s}^{m}(\beta y_{1}L(y_{2}))$, where $\beta \in \mathbb{F}_{p^m}^{*}$ and $L(x)=\sum a_{i}x^{q^{i}} \ (q=p^s)$ is a $q$-polynomial over $\mathbb{F}_{p^m}$ inducing a permutation of $\mathbb{F}_{p^m}$. Let $\gamma$ be a nonzero element of $\mathbb{F}_{p^m}$. Then $H: \mathbb{F}_{p^n} \times \mathbb{F}_{p^m} \times \mathbb{F}_{p^m}\rightarrow \mathbb{F}_{p^s}$ defined as
\begin{equation}\label{8}
  H(x, y_{1}, y_{2})=F_{Tr_{s}^{m}(\gamma y_{2}^{2})}(x)+G(y_{1}, y_{2})
\end{equation}
is a vectorial dual-bent function for which for any $c \in \mathbb{F}_{p^s}^{*}$, $(H_{c})^{*}=(H^{*})_{\sigma(c)}$, where $\sigma(c)=c^{-1}$, $H^{*}(x, y_{1}, y_{2})=R_{Tr_{s}^{m}(\gamma(L^{-1}(\beta^{-1}y_{1}))^{2})}(x)-Tr_{s}^{m}(L^{-1}(\beta^{-1}y_{1})$ $y_{2})$, $R_{0}(x)=Tr_{s}^{n}(-4^{-1}\alpha_{1}^{-1}x^{2})$, $R_{i}(x)=Tr_{s}^{n}(-4^{-1}\alpha_{2}^{-1}x^{2})$ if $i$ is a square in $\mathbb{F}_{p^s}^{*}$, $R_{i}(x)=Tr_{s}^{n}(-4^{-1}\alpha_{3}^{-1}x^{2})$ if $i$ is a non-square in $\mathbb{F}_{p^s}^{*}$.
\end{theorem}
\begin{proof}
For any $c \in \mathbb{F}_{p^s}^{*}$,
\begin{equation*}
  \begin{split}
  H_{c}(x, y_{1}, y_{2})
  &\triangleq Tr_{1}^{s}(c H(x, y_{1}, y_{2}))\\
  &=Tr_{1}^{s}(c F_{Tr_{s}^{m}(\gamma y_{2}^{2})}(x))+Tr_{1}^{m}(c\beta y_{1}L(y_{2})).
  \end{split}
\end{equation*}
Let $S$ and $N$ denote the sets of squares and non-squares in $\mathbb{F}_{p^s}^{*}$ respectively. For any $c \in \mathbb{F}_{p^s}^{*}$ and  $(a, b_{1}, b_{2}) \in \mathbb{F}_{p^n} \times \mathbb{F}_{p^m} \times \mathbb{F}_{p^m}$, we have
\begin{equation*}
  \begin{split}
  &W_{H_{c}}(a, b_{1}, b_{2})\\
  &=\sum_{x \in \mathbb{F}_{p^n}}\sum_{y_{1} \in \mathbb{F}_{p^m}}\sum_{y_{2} \in \mathbb{F}_{p^m}}\zeta_{p}^{Tr_{1}^{s}(c F_{Tr_{s}^{m}(\gamma y_{2}^{2})}(x))+Tr_{1}^{m}(c\beta y_{1}L(y_{2}))-Tr_{1}^{n}(ax)-Tr_{1}^{m}(b_{1}y_{1}+b_{2}y_{2})}\\
  &=\sum_{y_{2}: Tr_{s}^{m}(\gamma y_{2}^{2})=0}\zeta_{p}^{-Tr_{1}^{m}(b_{2}y_{2})}\sum_{x \in \mathbb{F}_{p^n}}\zeta_{p}^{Tr_{1}^{n}(c\alpha_{1}x^{2})-Tr_{1}^{n}(ax)}\sum_{y_{1} \in \mathbb{F}_{p^m}}\zeta_{p}^{Tr_{1}^{m}((c\beta L(y_{2})-b_{1})y_{1})}\\
  \end{split}
  \end{equation*}
  \begin{equation}\label{9}
  \begin{split}
  &+\sum_{y_{2}: Tr_{s}^{m}(\gamma y_{2}^{2})\in S}\zeta_{p}^{-Tr_{1}^{m}(b_{2}y_{2})}\sum_{x \in \mathbb{F}_{p^n}}\zeta_{p}^{Tr_{1}^{n}(c\alpha_{2}x^{2})-Tr_{1}^{n}(ax)}\sum_{y_{1} \in \mathbb{F}_{p^m}}\zeta_{p}^{Tr_{1}^{m}((c\beta L(y_{2})-b_{1})y_{1})}\\
  &+\sum_{y_{2}: Tr_{s}^{m}(\gamma y_{2}^{2})\in N}\zeta_{p}^{-Tr_{1}^{m}(b_{2}y_{2})}\sum_{x \in \mathbb{F}_{p^n}}\zeta_{p}^{Tr_{1}^{n}(c\alpha_{3}x^{2})-Tr_{1}^{n}(ax)}\sum_{y_{1} \in \mathbb{F}_{p^m}}\zeta_{p}^{Tr_{1}^{m}((c\beta L(y_{2})-b_{1})y_{1})}.
  \end{split}
\end{equation}
Note that $\sum_{y_{1} \in \mathbb{F}_{p^m}}\zeta_{p}^{Tr_{1}^{m}((c\beta L(y_{2})-b_{1})y_{1})}\neq 0$ if and only if $y_{2}=L^{-1}(c^{-1}\beta^{-1}b_{1})$ $=c^{-1}L^{-1}(\beta^{-1}b_{1})$. Then by (5) and (9), we have
\begin{equation}\label{10}
  \begin{split}
   &W_{H_{c}}(a, b_{1}, b_{2})\\
   &=\delta_{\{0\}}(\theta)\kappa_{1}p^{m+\frac{n}{2}}
   \zeta_{p}^{-Tr_{1}^{n}(c^{-1}4^{-1}\alpha_{1}^{-1}a^{2})-Tr_{1}^{m}(c^{-1}L^{-1}(\beta^{-1}b_{1})b_{2})}\\
   &\ \ +\delta_{S}(\theta)\kappa_{2}p^{m+\frac{n}{2}}
   \zeta_{p}^{-Tr_{1}^{n}(c^{-1}4^{-1}\alpha_{2}^{-1}a^{2})-Tr_{1}^{m}(c^{-1}L^{-1}(\beta^{-1}b_{1})b_{2})}\\
   &\ \ +\delta_{N}(\theta)\kappa_{3}p^{m+\frac{n}{2}}
   \zeta_{p}^{-Tr_{1}^{n}(c^{-1}4^{-1}\alpha_{3}^{-1}a^{2})-Tr_{1}^{m}(c^{-1}L^{-1}(\beta^{-1}b_{1})b_{2})},
  \end{split}
\end{equation}
where $\theta=Tr_{s}^{m}(\gamma(L^{-1}(\beta^{-1}b_{1}))^{2})$, $\kappa_{j}=(-1)^{n-1}\epsilon^{n}\eta_{n}(c\alpha_{j}), 1\leq j\leq 3$, $\epsilon=1$ if $p\equiv 1 \ (mod \ 4)$ and $\epsilon=\sqrt{-1}$ if $p\equiv 3 \ (mod \ 4)$, and for any set $A$, $\delta_{A}: A\rightarrow \{0, 1\}$ is defined by
\begin{equation*}
  \delta_{A}(a)=\left\{
  \begin{split}
  & 1, \ \text{if} \ a \in A,\\
  & 0, \ \text{if} \ a \notin A.
  \end{split}\right.
\end{equation*}
Therefore, by (10), we can see that for any $c \in \mathbb{F}_{p^s}^{*}$, $H_{c}$ is a bent function with dual
\begin{equation}\label{11}
  (H_{c})^{*}=J_{c^{-1}},
\end{equation}
where $J(x, y_{1}, y_{2})=R_{Tr_{s}^{m}(\gamma(L^{-1}(\beta^{-1}y_{1}))^{2})}(x)-Tr_{s}^{m}(L^{-1}(\beta^{-1}y_{1})y_{2})$, $R_{0}(x)=Tr_{s}^{n}(-4^{-1}\alpha_{1}^{-1}x^{2})$, $R_{i}(x)=Tr_{s}^{n}(-4^{-1}\alpha_{2}^{-1}x^{2})$ if $i$ is a square in $\mathbb{F}_{p^s}^{*}$, $R_{i}(x)=Tr_{s}^{n}(-4^{-1}\alpha_{3}^{-1}x^{2})$ if $i$ is a non-square in $\mathbb{F}_{p^s}^{*}$.

Since $J$ has the similar form as $H$, with the similar argument as above, for any $c \in \mathbb{F}_{p^s}^{*}$, $(a, b_{1}, b_{2})\in \mathbb{F}_{p^n}\times \mathbb{F}_{p^m}\times \mathbb{F}_{p^m}$ we obtain
\begin{equation}\label{12}
  \begin{split}
   &W_{J_{c}}(a, b_{1}, b_{2})\\
   &=\delta_{\{0\}}(\theta')\kappa_{1}'p^{m+\frac{n}{2}}
   \zeta_{p}^{Tr_{1}^{n}(c^{-1}\alpha_{1}a^{2})+Tr_{1}^{m}(c^{-1}\beta b_{1}L(b_{2}))}\\
   &\ \ +\delta_{S}(\theta')\kappa_{2}'p^{m+\frac{n}{2}}
    \zeta_{p}^{Tr_{1}^{n}(c^{-1}\alpha_{2}a^{2})+Tr_{1}^{m}(c^{-1}\beta b_{1}L(b_{2}))}\\
   &\ \ +\delta_{N}(\theta')\kappa_{3}'p^{m+\frac{n}{2}}
    \zeta_{p}^{Tr_{1}^{n}(c^{-1}\alpha_{3}a^{2})+Tr_{1}^{m}(c^{-1}\beta b_{1}L(b_{2}))},
  \end{split}
\end{equation}
where $\theta'=Tr_{s}^{m}(\gamma b_{2}^{2})$, $\kappa_{j}'=(-1)^{n-1}\epsilon^{n}\eta_{n}(-c4^{-1}\alpha_{j}^{-1}), 1\leq j\leq 3$. By (12), $|W_{J_{c}}(a, b_{1}, b_{2})|=p^{m+\frac{n}{2}}$ for any $c \in \mathbb{F}_{p^s}^{*}$, $(a, b_{1}, b_{2})\in \mathbb{F}_{p^n}\times \mathbb{F}_{p^m}\times \mathbb{F}_{p^m}$, hence $J$ is vectorial bent. By (11), we have that $H$ is vectorial dual-bent for which $J$ is a vectorial dual of $H$.
\qed
\end{proof}

We give an example of vectorial dual-bent function by using Theorem 1.

\begin{example}\label{Example1}
Let $w$ be a primitive element of $\mathbb{F}_{3^6}$. Let $F_{i} \ (i \in \mathbb{F}_{3^2}): \mathbb{F}_{3^6}\rightarrow \mathbb{F}_{3^2}$ be defined as $F_{0}(x)=Tr_{2}^{6}(x^{2})$, $F_{i}(x)=Tr_{2}^{6}(wx^{2})$ if $i$ is a square in $\mathbb{F}_{3^2}^{*}$, $F_{i}(x)=Tr_{2}^{6}(w^{2}x^{2})$ if $i$ is a non-square in $\mathbb{F}_{3^2}^{*}$. Then by Theorem 1, $H(x, y_{1}, y_{2})=F_{Tr_{2}^{4}(y_{2}^{2})}(x)+Tr_{2}^{4}(y_{1}y_{2})$ is a vectorial dual-bent function from $\mathbb{F}_{3^6}\times \mathbb{F}_{3^4} \times \mathbb{F}_{3^4}$ to $\mathbb{F}_{3^2}$.
\end{example}
\section{New results on vectorial dual-bent functions and partial difference sets}
\label{sec:4}
In this section, we provide new results on vectorial dual-bent functions and partial difference sets, and we illustrate that the known results on constructing partial difference sets by using bent functions and vectorial bent functions given in Proposition 1 (ii), (iii) and Proposition 2 (iii)-(v) can be obtained by our results.

In the sequel, for a given vectorial function $F: V_{n}\rightarrow V_{s}$, let
\begin{equation}\label{13}
  D_{i}=\{x \in V_{n}: F(x)=i\}, i \in V_{s}.
\end{equation}
Note that if $F(0)=0$, then $D_{0}^{F}=D_{0} \backslash \{0\}$. The following proposition gives a formula to compute $\chi_{u}(D_{i})$ for any $u \in V_{n}$ and $i \in V_{s}$, where $\chi_{u}$ denotes the character $\chi_{u}(x)=\zeta_{p}^{\langle u, x\rangle_{n}}, x \in V_{n}$.
\begin{proposition}\label{Proposition3}
Let $F: V_{n}\rightarrow V_{s}$ and for any $c \in V_{s}^{*}$, let $F_{c}$ denote the component function of $F$, that is, $F_{c}(x)=\langle c, F(x)\rangle_{s}$. Then
\begin{equation*}
  |D_{i}|=\chi_{0}(D_{i})=p^{n-s}+p^{-s}\sum_{c \in V_{s}^{*}}W_{F_{c}}(0)\zeta_{p}^{-\langle c, i\rangle_{s}}, \ i \in V_{s},
\end{equation*}
and
\begin{equation*}
  \chi_{u}(D_{i})=p^{-s}\sum_{c \in V_{s}^{*}}W_{F_{c}}(-u)\zeta_{p}^{-\langle c, i\rangle_{s}}, \ u \in  V_{n}^{*}, i \in V_{s}.
\end{equation*}
\end{proposition}
\begin{proof}
For any $u \in V_{n}$ and $i \in V_{s}$, we have
\begin{equation*}
\begin{split}
   \chi_{u}(D_{i}) & = \sum_{x \in D_{i}}\zeta_{p}^{\langle u, x\rangle_{n}}\\
                   & =p^{-s}\sum_{x \in V_{n}}\zeta_{p}^{\langle u, x \rangle_{n}}\sum_{c \in V_{s}}\zeta_{p}^{\langle c, F(x)-i\rangle_{s}}\\
                   & =p^{-s}\sum_{c \in V_{s}}\sum_{x \in V_{n}}\zeta_{p}^{\langle c, F(x)\rangle_{s}+\langle u, x \rangle_{n}}\zeta_{p}^{-\langle c,i\rangle_{s}}\\
                   & =p^{n-s}\delta(u, 0)+p^{-s}\sum_{c \in V_{s}^{*}}W_{F_{c}}(-u)\zeta_{p}^{-\langle c,i\rangle_{s}},
                  \end{split}
\end{equation*}
where $\delta (u, 0)=1$ if $u=0$ and $\delta (u, 0)=0$ if $u\neq 0$.
\qed
\end{proof}

Lemma 3 of \cite{Cesmelioglu2} is a special case of Proposition 3 with $i=0$. By Proposition 3, we give the following corollary.
\begin{corollary}\label{Corollary1}
Let $n$ be a positive even integer. Let $F: V_{n}\rightarrow \mathbb{F}_{p^s}$ be a vectorial dual-bent function which satisfies that $F(0)=0, F(-x)=F(x)$ and all component functions are regular (that is, $\varepsilon_{F_{c}}=1$ for all $c \in \mathbb{F}_{p^s}^{*}$) or weakly regular but not regular (that is, $\varepsilon_{F_{c}}=-1$ for all $c \in \mathbb{F}_{p^s}^{*}$). Suppose $\varepsilon_{F_{c}}=\varepsilon, c \in \mathbb{F}_{p^s}^{*}$. Then
\begin{equation*}
  |D_{0}|=p^{n-s}+\varepsilon(p^s-1)p^{\frac{n}{2}-s}, |D_{i}|=p^{n-s}-\varepsilon p^{\frac{n}{2}-s}, i \in \mathbb{F}_{p^s}^{*}.
\end{equation*}
\end{corollary}
\begin{proof}
Let $F^{*}$ be a vectorial dual of $F$, then $(F_{c})^{*}=(F^{*})_{\sigma(c)}, c \in \mathbb{F}_{p^s}^{*}$ for a permutation $\sigma$ over $\mathbb{F}_{p^s}^{*}$. By the assumption for $F$, for any $c \in \mathbb{F}_{p^s}^{*}$ we have
 \begin{equation*}
   W_{F_{c}}(0)=\varepsilon p^{\frac{n}{2}}\zeta_{p}^{(F_{c})^{*}(0)}=\varepsilon p^{\frac{n}{2}}\zeta_{p}^{(F^{*})_{\sigma(c)}(0)},
\end{equation*}
where $\varepsilon \in \{\pm 1\}$ is a constant independent of $c$. By Corollary 2 and Proposition 5 of \cite{Cesmelioglu2}, $F^{*}(0)=0$. Therefore, by Proposition 3, we have
\begin{equation*}
  \begin{split}
    |D_{i}| &=p^{n-s}+p^{-s}\sum_{c \in \mathbb{F}_{p^s}^{*}}W_{F_{c}}(0)\zeta_{p}^{-Tr_{1}^{s}(ci)}\\
            &=p^{n-s}+\varepsilon p^{\frac{n}{2}-s}\sum_{c \in \mathbb{F}_{p^s}^{*}}\zeta_{p}^{-Tr_{1}^{s}(ci)}\\
            &=p^{n-s}+\varepsilon p^{\frac{n}{2}-s}(p^s \delta(i, 0)-1),
   \end{split}
\end{equation*}
where $\delta (i, 0)=1$ if $i=0$ and $\delta (i, 0)=0$ if $i\neq 0$.
\qed
\end{proof}

For a vectorial dual-bent function $F: V_{n}\rightarrow \mathbb{F}_{p^s}$ satisfying the condition of Corollary 1, if the permutation $\sigma$ of $\mathbb{F}_{p^s}^{*}$ given by $(F_{c})^{*}=(F^{*})_{\sigma(c)}, c \in \mathbb{F}_{p^s}^{*}$ (where $F^{*}$ is a vectorial dual of $F$) is an identity map, we show that the preimage set of any subset of $\mathbb{F}_{p^s}$ for $F$ forms a partial difference set.

\begin{theorem}\label{Theorem2}
Let $n$ be a positive even integer. Let $F: V_{n}\rightarrow \mathbb{F}_{p^s}$ be a vectorial dual-bent function which satisfies that $F(0)=0, F(-x)=F(x)$ and all component functions are regular (that is, $\varepsilon_{F_{c}}=1$ for all $c \in \mathbb{F}_{p^s}^{*}$) or weakly regular but not regular (that is, $\varepsilon_{F_{c}}=-1$ for all $c \in \mathbb{F}_{p^s}^{*}$). Suppose $\varepsilon_{F_{c}}=\varepsilon, c \in \mathbb{F}_{p^s}^{*}$. Let $F^{*}$ be a vectorial dual of $F$ and $(F_{c})^{*}=(F^{*})_{\sigma(c)}, c \in \mathbb{F}_{p^s}^{*}$ for a permutation $\sigma$ over $\mathbb{F}_{p^s}^{*}$. If $\sigma$ is an identity map, that is, $\sigma(c)=c, c \in \mathbb{F}_{p^s}^{*}$, then for any subset $A$ of $\mathbb{F}_{p^s}$, $D_{A}^{F}=\{x \in V_{n}^{*}: F(x) \in A\}$ is a $(p^n, k, \lambda, \mu)$ partial difference set, where

\emph{(i)} if $0 \in A$,
\begin{equation*}
  \begin{split}
      & k=|A|p^{n-s}+\varepsilon (p^s-|A|) p^{\frac{n}{2}-s}-1,\\
      & \lambda=p^{n-2s}|A|^{2}+\varepsilon(p^s-|A|)p^{\frac{n}{2}-s}-2,\\
      & \mu=p^{n-2s}|A|^{2}+\varepsilon|A|p^{\frac{n}{2}-s};
  \end{split}
\end{equation*}

\emph{(ii)} if $0 \notin A$,
\begin{equation*}
  \begin{split}
       & k=|A|p^{n-s}-\varepsilon|A|p^{\frac{n}{2}-s}, \\
       & \lambda=p^{n-2s}|A|^{2}+\varepsilon(p^s-3|A|)p^{\frac{n}{2}-s},\\
       & \mu=p^{n-2s}|A|^{2}-\varepsilon|A|p^{\frac{n}{2}-s}.
   \end{split}
\end{equation*}
\end{theorem}
\begin{proof}
Let $D=\{x \in V_{n}: F(x)\in A\}$. For any $u \in V_{n}^{*}$, by Proposition 3 and the assumption for $F$,
\begin{equation*}
\begin{split}
    \chi_{u}(D) & =\sum_{i \in A}\chi_{u}(D_{i}) \\
                & =p^{-s}\sum_{i \in A}\sum_{c \in \mathbb{F}_{p^s}^{*}}W_{F_{c}}(-u)\zeta_{p}^{Tr_{1}^{s}(-ci)}\\
                & =\varepsilon p^{\frac{n}{2}-s}\sum_{i \in A}\sum_{c \in \mathbb{F}_{p^s}^{*}}\zeta_{p}^{Tr_{1}^{s}(cF^{*}(-u)-ci)}\\
                & =\left\{\begin{split}
                               \varepsilon p^{\frac{n}{2}}-\varepsilon p^{\frac{n}{2}-s}|A|, & \ \ \text{if} \ F^{*}(-u) \in A, \\
                               -\varepsilon p^{\frac{n}{2}-s}|A|, & \ \ \text{if} \ F^{*}(-u) \notin A.
                          \end{split}\right.
              \end{split}
\end{equation*}
(i)  If $0 \in A$, then $D_{A}^{F}=D \backslash \{0\}$ and $\chi_{u}(D_{A}^{F})=\chi_{u}(D)-1$. By Corollary 1, $|D_{A}^{F}|=(|A|-1)(p^{n-s}-\varepsilon p^{\frac{n}{2}-s})+(p^{n-s}+\varepsilon(p^s-1)p^{\frac{n}{2}-s}-1)=|A|p^{n-s}+\varepsilon(p^s-|A|)p^{\frac{n}{2}-s}-1$. By Lemma 1, $D_{A}^{F}$ is a $(p^n, k, \lambda, \mu)$ partial difference set with
\begin{equation*}
  \begin{split}
      & k=|A|p^{n-s}+\varepsilon (p^s-|A|) p^{\frac{n}{2}-s}-1,\\
      & \lambda=p^{n-2s}|A|^{2}+\varepsilon(p^s-|A|)p^{\frac{n}{2}-s}-2,\\
      & \mu=p^{n-2s}|A|^{2}+\varepsilon|A|p^{\frac{n}{2}-s}.
  \end{split}
\end{equation*}
(ii) If $0 \notin A$, then $D_{A}^{F}=D$ and $\chi_{u}(D_{A}^{F})=\chi_{u}(D)$. By Corollary 1, $|D_{A}^{F}|=|A|(p^{n-s}-\varepsilon p^{\frac{n}{2}-s})$. By Lemma 1, $D_{A}^{F}$ is a $(p^n, k, \lambda, \mu)$ partial difference set with
\begin{equation*}
  \begin{split}
       & k=|A|p^{n-s}-\varepsilon |A|p^{\frac{n}{2}-s}, \\
       & \lambda=p^{n-2s}|A|^{2}+\varepsilon (p^s-3|A|)p^{\frac{n}{2}-s},\\
       & \mu=p^{n-2s}|A|^{2}-\varepsilon |A|p^{\frac{n}{2}-s}.
   \end{split}
\end{equation*}
\qed
\end{proof}

\begin{remark}\label{Remark1}
Let $n$ be a positive even integer. Let $f: V_{n}\rightarrow \mathbb{F}_{p}$ be a weakly regular bent function of $l$-form with $f(0)=0$ and $gcd(l-1, p-1)=1$ (which implies that $f(-x)=f(x)$ and $cf$ are weakly regular bent with $\varepsilon_{cf}=\varepsilon_{f}$ for all $c \in \mathbb{F}_{p}^{*}$). Then by Proposition 2 and its proof of \cite{Cesmelioglu1}, $f$ (seen as a vectorial bent function from $V_{n}$ to $V_{1}$) is vectorial dual-bent with $(cf)^{*}=\sigma(c)f^{*}, \sigma(c)=c^{1-t}, c \in \mathbb{F}_{p}^{*}$, where $(l-1)(t-1)\equiv 1 \ mod \ (p-1)$. When $p=3$ and $l=2$, then $\sigma(c)=c^{-1}, c \in \mathbb{F}_{3}^{*}$. Note that $c^{-1}=c, c \in \mathbb{F}_{3}^{*}$, that is, $\sigma$ is an identity map over $\mathbb{F}_{3}^{*}$. Hence, in this case, $D_{0} \ \backslash \ \{0\}, D_{1}, D_{2}$ for $f$ are all partial difference sets by Theorem 2, which is the main result given in \cite{Tan}.
\end{remark}
\begin{remark}\label{Remark2}
When $s\geq 2$, to the best of our knowledge, the known vectorial dual-bent functions $F: V_{n}\rightarrow V_{s}$ such that the permutation $\sigma$ of $V_{s}^{*}$ given by $(F_{c})^{*}=(F^{*})_{\sigma(c)}, c \in V_{s}^{*}$ (where $F^{*}$ is a vectorial dual of $F$) is an identity map are vectorial partial spread bent functions. For vectorial partial spread bent functions, the above theorem is clearly true by Theorem 2.2 of \cite{Ma}. An interesting open problem is that whether there exist vectorial dual-bent functions $F$ which are not linear equivalent to vectorial partial spread bent functions such that the permutation $\sigma$ is an identity map.
\end{remark}

For a vectorial dual-bent function $F: V_{n}\rightarrow \mathbb{F}_{p^s}$ satisfying the condition of Corollary 1, if the permutation $\sigma$ of $\mathbb{F}_{p^s}^{*}$ given by $(F_{c})^{*}=(F^{*})_{\sigma(c)}, c \in \mathbb{F}_{p^s}^{*}$ (where $F^{*}$ is a vectorial dual of $F$) is not an identity map, but it satisfies $\sigma^{-1}(c)H_{l}=cH_{l}, c \in \mathbb{F}_{p^s}^{*}$ for some subgroup $H_{l}=\{x^{l}, x \in \mathbb{F}_{p^s}^{*}\}\subseteq \mathbb{F}_{p^s}^{*}$, we show that any union of sets from $\{D_{0}^{F}, D_{\beta H_{l}}^{F}, \beta \in \mathbb{F}_{p^s}^{*}\}$ is a partial difference set.

\begin{theorem}\label{Theorem3}
Let $n$ be a positive even integer. Let $F: V_{n}\rightarrow \mathbb{F}_{p^s}$ be a vectorial dual-bent function which satisfies that $F(0)=0, F(-x)=F(x)$ and all component functions are regular (that is, $\varepsilon_{F_{c}}=1$ for all $c \in \mathbb{F}_{p^s}^{*}$) or weakly regular but not regular (that is, $\varepsilon_{F_{c}}=-1$ for all $c \in \mathbb{F}_{p^s}^{*}$). Suppose $\varepsilon_{F_{c}}=\varepsilon, c \in \mathbb{F}_{p^s}^{*}$. Let $F^{*}$ be a vectorial dual of $F$ and $(F_{c})^{*}=(F^{*})_{\sigma(c)}, c \in \mathbb{F}_{p^s}^{*}$ for a permutation $\sigma$ over $\mathbb{F}_{p^s}^{*}$. If $\sigma$ satisfies that $\sigma^{-1}(c)H_{l}=cH_{l}, c \in \mathbb{F}_{p^s}^{*}$ for some subgroup $H_{l}=\{x^{l}, x \in \mathbb{F}_{p^s}^{*}\}\subseteq \mathbb{F}_{p^s}^{*}$, then for any $\beta \in \mathbb{F}_{p^s}^{*}$, $D_{\beta H_{l}}^{F}=\{x \in V_{n}^{*}: F(x) \in \beta H_{l}\}$ is a $(p^n, k, \lambda, \mu)$ partial difference set, where
\begin{equation*}
  \begin{split}
      & k=\frac{p^s-1}{gcd(l, p^s-1)}p^{n-s}-\varepsilon p^{\frac{n}{2}-s}\frac{p^s-1}{gcd(l, p^s-1)},\\
      & \lambda=(\frac{p^s-1}{gcd(l, p^s-1)})^{2}p^{n-2s}+\varepsilon p^{\frac{n}{2}-s}(p^s-3\frac{p^s-1}{gcd(l, p^s-1)}),\\
  \end{split}
\end{equation*}
  \begin{equation}\label{14}
  \begin{split}
      \mu=(\frac{p^s-1}{gcd(l, p^s-1)})^{2}p^{n-2s}-\varepsilon p^{\frac{n}{2}-s}\frac{p^s-1}{gcd(l, p^s-1)}.
  \end{split}
\end{equation}
Furthermore, any union of sets from $\{D_{0}^{F}, D_{\beta H_{l}}^{F}, \beta \in \mathbb{F}_{p^s}^{*}\}$, denoted by $A$, is a $(p^n, k, \lambda, \mu)$ partial difference set, where
\begin{equation}\label{15}
  \begin{split}
      & k=(m_{1}\frac{p^s-1}{gcd(l, p^s-1)}+m_{0})p^{n-s}+\varepsilon p^{\frac{n}{2}-s}(m_{0}p^s-(m_{1}\frac{p^s-1}{gcd(l, p^s-1)}+m_{0}))\\
      & \ \ \ \ \ -m_{0},\\
      & \lambda=(m_{1}\frac{p^s-1}{gcd(l, p^s-1)}+m_{0})^{2}p^{n-2s}+\varepsilon p^{\frac{n}{2}-s}(p^s+(2m_{0}-3)m_{1}\frac{p^s-1}{gcd(l, p^s-1)}\\
      & \ \ \ \ \ -m_{0})-2m_{0},\\
      & \mu=(m_{1}\frac{p^s-1}{gcd(l, p^s-1)}+m_{0})^{2}p^{n-2s}+\varepsilon p^{\frac{n}{2}-s}((2m_{0}-1)m_{1}\frac{p^s-1}{gcd(l, p^s-1)}\\
      & \ \ \ \ \ +m_{0}),
  \end{split}
\end{equation}
$0\leq m_{1}\leq gcd(l, p^s-1)$ is the number of distinct $D_{\beta H_{l}}^{F}$ in $A$, $m_{0}=1$ if $ D_{0}^{F} \subseteq A$ and $m_{0}=0$ if $D_{0}^{F} \not \subseteq A$. In particular, if $\sigma(c)=c^{-t}, c \in \mathbb{F}_{p^s}^{*}$ for some $t$ with $gcd(t, p^s-1)=1$, let $l$ satisfy $gcd(l, p^s-1)\mid (1+r)$, where $tr\equiv 1 \ mod \ (p^s-1)$, then $D_{\beta H_{l}}^{F}, \beta \in \mathbb{F}_{p^s}^{*}$ are $(p^n, k, \lambda, \mu)$ partial difference sets, where $k, \lambda, \mu$ are given in (14), and any union of sets from $\{D_{0}^{F}, D_{\beta H_{l}}^{F}, \beta \in \mathbb{F}_{p^s}^{*}\}$ is a $(p^n, k, \lambda, \mu)$ partial difference set, where $k, \lambda, \mu$ are given in (15).
\end{theorem}
\begin{proof}
For any $\beta \in \mathbb{F}_{p^s}^{*}$, $u \in V_{n}^{*}$, by Proposition 3 and the assumption for $F$,
\begin{equation*}
\begin{split}
    \chi_{u}(D_{\beta H_{l}}^{F}) & =\sum_{i \in \beta H_{l}}\chi_{u}(D_{i}) \\
                & =p^{-s}\sum_{i \in \beta H_{l}}\sum_{c \in \mathbb{F}_{p^s}^{*}}W_{F_{c}}(-u)\zeta_{p}^{Tr_{1}^{s}(-ci)}\\
                & =\varepsilon p^{\frac{n}{2}-s}\sum_{i \in \beta H_{l}}\sum_{c \in \mathbb{F}_{p^s}^{*}}\zeta_{p}^{Tr_{1}^{s}(\sigma(c)F^{*}(-u)-ci)}\\
                & =\varepsilon p^{\frac{n}{2}-s}\sum_{c \in \mathbb{F}_{p^s}^{*}}\zeta_{p}^{Tr_{1}^{s}(\sigma(c)F^{*}(-u))}\sum_{i \in \beta H_{l}}\zeta_{p}^{Tr_{1}^{s}(-ci)}\\
                & =\varepsilon p^{\frac{n}{2}-s}\sum_{c \in \mathbb{F}_{p^s}^{*}}\zeta_{p}^{Tr_{1}^{s}(\sigma(c)F^{*}(-u))}\sum_{i \in H_{l}}\zeta_{p}^{Tr_{1}^{s}(-c \beta i)}\\
                & =\varepsilon p^{\frac{n}{2}-s}\sum_{c \in \mathbb{F}_{p^s}^{*}}\zeta_{p}^{Tr_{1}^{s}(c F^{*}(-u))}\sum_{i \in H_{l}}\zeta_{p}^{Tr_{1}^{s}(-\sigma^{-1}(c) \beta i)}\\
                & =\varepsilon p^{\frac{n}{2}-s}\sum_{c \in \mathbb{F}_{p^s}^{*}}\zeta_{p}^{Tr_{1}^{s}(c F^{*}(-u))}\sum_{i \in \sigma^{-1}(c)H_{l}}\zeta_{p}^{Tr_{1}^{s}(-\beta i)}\\
                & =\varepsilon p^{\frac{n}{2}-s}\sum_{c \in \mathbb{F}_{p^s}^{*}}\zeta_{p}^{Tr_{1}^{s}(c F^{*}(-u))}\sum_{i \in cH_{l}}\zeta_{p}^{Tr_{1}^{s}(-\beta i)}\\
\end{split}
\end{equation*}
\begin{equation}\label{16}
\begin{split}
                & =\varepsilon p^{\frac{n}{2}-s}\sum_{c \in \mathbb{F}_{p^s}^{*}}\zeta_{p}^{Tr_{1}^{s}(c F^{*}(-u))}\sum_{i \in H_{l}}\zeta_{p}^{Tr_{1}^{s}(-\beta ci)}\\
                & =\varepsilon p^{\frac{n}{2}-s}\sum_{i \in H_{l}}\sum_{c \in \mathbb{F}_{p^s}^{*}}\zeta_{p}^{Tr_{1}^{s}(c(F^{*}(-u)-\beta i))}\\
                & =\left\{\begin{split}
                \varepsilon p^{\frac{n}{2}}-\varepsilon p^{\frac{n}{2}-s}|H_{l}|, & \ \ \text{if} \ F^{*}(-u) \in \beta H_{l}, \\
                               -\varepsilon p^{\frac{n}{2}-s}|H_{l}|, & \ \ \text{if} \ F^{*}(-u) \notin \beta H_{l},
                \end{split}\right.
    \end{split}
\end{equation}
where in the eighth equation we use the property that $\sigma^{-1}(c)H_{l}=cH_{l}, c \in \mathbb{F}_{p^s}^{*}$. By Corollary 1, Lemma 1 and (16), $D_{\beta H_{l}}^{F}$ is a $(p^n, k, \lambda, \mu)$ partial difference set with
\begin{equation*}
  \begin{split}
      & k=|H_{l}|p^{n-s}-\varepsilon p^{\frac{n}{2}-s}|H_{l}|,\\
      & \lambda=|H_{l}|^{2}p^{n-2s}+\varepsilon p^{\frac{n}{2}-s}(p^s-3|H_{l}|),\\
      & \mu=|H_{l}|^{2}p^{n-2s}-\varepsilon p^{\frac{n}{2}-s}|H_{l}|.
  \end{split}
\end{equation*}
Note that $|H_{l}|=\frac{p^s-1}{gcd(l, p^s-1)}$. Let $w$ denote a primitive element of $\mathbb{F}_{p^s}$. For any $u \in V_{n}^{*}$ and $0\leq i_{1}< \dots < i_{m_{1}}\leq [\mathbb{F}_{p^s}^{*}: H_{l}]-1=gcd(l, p^s-1)-1$, by (16),
\begin{equation}\label{17}
\begin{split}
           &\sum_{j=1}^{m_{1}}\chi_{u}(D_{w^{i_{j}}H_{l}}^{F})\\
           & = \left\{\begin{split}
                \varepsilon p^{\frac{n}{2}}-\varepsilon p^{\frac{n}{2}-s}m_{1}|H_{l}|, & \ \ \text{if} \ F^{*}(-u) \in w^{i_{j_{0}}} H_{l} \ \text{for some} \ 1\leq j_{0} \leq m_{1}, \\
                -\varepsilon p^{\frac{n}{2}-s}m_{1}|H_{l}|, & \ \ \text{if} \ F^{*}(-u) \notin \cup_{j=1}^{m_{1}}w^{i_{j}} H_{l}.
                \end{split}\right.
    \end{split}
\end{equation}
By Proposition 4 of \cite{Cesmelioglu2},
\begin{equation}\label{18}
  \chi_{u}(D_{0}^{F})=\left\{
  \begin{split}
  \varepsilon p^{\frac{n}{2}-s}(p^s-1)-1, & \ \ \text{if} \ F^{*}(-u)=0,\\
  -\varepsilon p^{\frac{n}{2}-s}-1, & \ \ \text{if} \ F^{*}(-u)\neq 0.
  \end{split}\right.
\end{equation}
Then
\begin{equation}\label{19}
\begin{split}
  &\sum_{j=1}^{m_{1}}\chi_{u}(D_{w^{i_{j}}H_{l}}^{F})+\chi_{u}(D_{0}^{F})\\
  &=\left\{
  \begin{split}
  \varepsilon p^{\frac{n}{2}-s}(p^s-1-m_{1}|H_{l}|)-1, & \ \ \text{if} \ F^{*}(-u)=0 \ \text{or} \  F^{*}(-u) \in w^{i_{j_{0}}} H_{l}\\
  & \ \ \text{for some} \ 1\leq j_{0} \leq m_{1},\\
  \varepsilon p^{\frac{n}{2}-s}(-1-m_{1}|H_{l}|)-1, & \ \ \text{if} \ F^{*}(-u)\neq 0 \ \text{and} \  F^{*}(-u) \notin \cup_{j=1}^{m_{1}}w^{i_{j}}H_{l}.
  \end{split}\right.
\end{split}
\end{equation}
By Corollary 1, Lemma 1 and (17)-(19), any union of sets from $\{D_{0}^{F}, D_{\beta H_{l}}^{F}, \beta \in \mathbb{F}_{p^s}^{*}\}$, denoted by $A$, is a $(p^n, k, \lambda, \mu)$ partial difference set, where
\begin{equation*}
  \begin{split}
      & k=(m_{1}|H_{l}|+m_{0})p^{n-s}+\varepsilon p^{\frac{n}{2}-s}(m_{0}p^s-(m_{1}|H_{l}|+m_{0}))-m_{0},\\
      & \lambda=(m_{1}|H_{l}|+m_{0})^{2}p^{n-2s}+\varepsilon p^{\frac{n}{2}-s}(p^s+(2m_{0}-3)m_{1}|H_{l}|-m_{0})-2m_{0},\\
      & \mu=(m_{1}|H_{l}|+m_{0})^{2}p^{n-2s}+\varepsilon p^{\frac{n}{2}-s}((2m_{0}-1)m_{1}|H_{l}|+m_{0}),
  \end{split}
\end{equation*}
$0\leq m_{1}\leq gcd(l, p^s-1)$ is the number of distinct $D_{\beta H_{l}}^{F}$ in $A$, $m_{0}=1$ if $ D_{0}^{F} \subseteq A$ and $m_{0}=0$ if $D_{0}^{F} \not \subseteq A$.

In particular, if $\sigma(c)=c^{-t}, c \in \mathbb{F}_{p^s}^{*}$ for some $t$ with $gcd(t, p^s-1)=1$, then $\sigma^{-1}(c)=c^{-r}, c \in \mathbb{F}_{p^s}^{*}$, where $tr\equiv 1 \ mod \ (p^s-1)$. In this case, it is easy to see that
\begin{equation*}
\begin{split}
  \sigma^{-1}(c)H_{l}=cH_{l} \ \text{for any} \ c \in \mathbb{F}_{p^s}^{*}
  &\Leftrightarrow c^{1+r} \in H_{l} \ \text{for any} \ c \in \mathbb{F}_{p^s}^{*}\\
  &\Leftrightarrow lx\equiv 1+r \ mod \ (p^s-1) \ \text{has a solution}\\
  &\Leftrightarrow gcd(l, p^s-1)\mid (1+r).
  \end{split}
\end{equation*}
\qed
\end{proof}
\begin{remark}\label{Remark3}
Let $n$ be a positive even integer. Let $f: V_{n}\rightarrow \mathbb{F}_{p}$ be a weakly regular bent function of $l$-form with $f(0)=0$ and $gcd(l-1, p-1)=1$, where $l\neq p-1$ when $p>3$ (which implies that $f(-x)=f(x)$ and $cf$ are weakly regular bent with $\varepsilon_{cf}=\varepsilon_{f}$ for all $c \in \mathbb{F}_{p}^{*}$). Then $f$ (seen as a vectorial bent function from $V_{n}$ to $V_{1}$) is vectorial dual-bent with $(cf)^{*}=\sigma(c)f^{*}, \sigma(c)=c^{1-t}, c \in \mathbb{F}_{p}^{*}$, where $(l-1)(t-1)\equiv 1 \ mod \ (p-1)$. Since $\sigma$ satisfies the condition of Theorem 3, we have that $D_{\beta H_{l}}^{f}, \beta \in \mathbb{F}_{p}^{*}$ are partial difference sets and any union of sets from $\{D_{0}^{f}, D_{\beta H_{l}}^{f}, \beta \in \mathbb{F}_{p}^{*}\}$ is a partial difference set, which is the result given in Proposition 1 (iii).
\end{remark}

Based on Theorem 3, we obtain the following corollary.

\begin{corollary}\label{Corollary2}
Let $n$ be a positive even integer. Let $F: V_{n}\rightarrow \mathbb{F}_{p^s}$ be a vectorial dual-bent function which satisfies that $F(0)=0, F(-x)=F(x)$ and all component functions are regular (that is, $\varepsilon_{F_{c}}=1$ for all $c \in \mathbb{F}_{p^s}^{*}$) or weakly regular but not regular (that is, $\varepsilon_{F_{c}}=-1$ for all $c \in \mathbb{F}_{p^s}^{*}$). Suppose $\varepsilon_{F_{c}}=\varepsilon, c \in \mathbb{F}_{p^s}^{*}$. Let $F^{*}$ be a vectorial dual of $F$ and $(F_{c})^{*}=(F^{*})_{\sigma(c)}, c \in \mathbb{F}_{p^s}^{*}$ for a permutation $\sigma$ over $\mathbb{F}_{p^s}^{*}$. Let $S$ (respectively, $N$) denote the set of all squares (respectively, non-squares) in $\mathbb{F}_{p^s}^{*}$. If $\sigma$ satisfies that $\sigma(S)=S$, then $D_{S}^{F}=\{x \in V_{n}^{*}: F(x) \in S\}$ and $D_{N}^{F}=\{x \in V_{n}^{*}: F(x) \in N\}$ are both $(p^n, k, \lambda, \mu)$ partial difference sets, where
\begin{equation*}
  \begin{split}
      & k=\frac{p^s-1}{2}(p^{n-s}-\varepsilon p^{\frac{n}{2}-s}),\\
      & \lambda=\frac{(p^s-1)^{2}}{4}p^{n-2s}-\varepsilon p^{\frac{n}{2}-s}\frac{p^s-3}{2},\\
      & \mu=\frac{(p^s-1)^{2}}{4}p^{n-2s}-\varepsilon p^{\frac{n}{2}-s}\frac{p^s-1}{2}.
  \end{split}
\end{equation*}
Furthermore, any union of sets from $\{D_{0}^{F}, D_{S}^{F}, D_{N}^{F}\}$, denoted by $A$, is a $(p^n, k,$ $ \lambda, \mu)$ partial difference set, where
\begin{equation*}
  \begin{split}
      & k=(m_{1}\frac{p^s-1}{2}+m_{0})p^{n-s}+\varepsilon p^{\frac{n}{2}-s}(m_{0}p^s-(m_{1}\frac{p^s-1}{2}+m_{0}))-m_{0},\\
      & \lambda=(m_{1}\frac{p^s-1}{2}+m_{0})^{2}p^{n-2s}+\varepsilon p^{\frac{n}{2}-s}(p^s+(2m_{0}-3)m_{1}\frac{p^s-1}{2}-m_{0})-2m_{0},\\
      & \mu=(m_{1}\frac{p^s-1}{2}+m_{0})^{2}p^{n-2s}+\varepsilon p^{\frac{n}{2}-s}((2m_{0}-1)m_{1}\frac{p^s-1}{2}+m_{0}),
  \end{split}
\end{equation*}
$0\leq m_{1}\leq 2$ is the number of $D_{S}^{F}, D_{N}^{F}$ in $A$, $m_{0}=1$ if $ D_{0}^{F} \subseteq A$ and $m_{0}=0$ if $D_{0}^{F} \not \subseteq A$.
\end{corollary}
\begin{proof}
Note that $S=H_{2}$, where $H_{2}=\{x^{2}: x \in \mathbb{F}_{p^s}^{*}\}$. We claim that
\begin{equation*}
  \sigma^{-1}(c)S=c S \ \text{for any} \ c \in \mathbb{F}_{p^s}^{*}\Leftrightarrow \sigma(S)=S.
\end{equation*}

($\Rightarrow$) For any $x^{2} \in S$, since $x^{2}S=S$, we have $\sigma^{-1}(x^{2})S=S$, which implies that $\sigma^{-1}(x^{2}) \in S$ and thus $\sigma^{-1}(S)\subseteq S$. Further, since $\sigma$ is a permutation, we have $\sigma^{-1}(S)=S$, therefore, $\sigma(S)=S$.

($\Leftarrow$) By $\sigma(S)=S$, we have $\sigma^{-1}(S)=S$, $\sigma^{-1}(N)=N$. For any $c \in \mathbb{F}_{p^s}^{*}$, if $c \in S$, then $c^{-1} \in S$, $\sigma^{-1}(c)\in S$, and thus $c^{-1}\sigma^{-1}(c)\in S$, that is, $\sigma^{-1}(c)S=c S$; if $c \notin S$, then $c^{-1} \in N$, $\sigma^{-1}(c)\in N$, and thus $c^{-1}\sigma^{-1}(c)\in S$, that is, $\sigma^{-1}(c)S=c S$.

By Theorem 3, $D_{S}^{F}=D_{H_{2}}^{F}$ and $D_{N}^{F}=D_{\beta H_{2}}^{F}$ are both partial difference sets, where $\beta \in \mathbb{F}_{p^s}^{*}$ is an arbitrary non-square, and any union of sets from $\{D_{0}^{F}, D_{S}^{F}, D_{N}^{F}\}$ is a partial difference set.
\qed
\end{proof}

\begin{remark}\label{Remark4}
Let $n$ be a positive even integer. Let $f: V_{n}\rightarrow \mathbb{F}_{p}$ be a weakly regular bent function of $l$-form with $f(0)=0$ and $gcd(l-1, p-1)=1$ (which implies that $f(-x)=f(x)$ and $cf$ are weakly regular bent with $\varepsilon_{cf}=\varepsilon_{f}$ for all $c \in \mathbb{F}_{p}^{*}$). Then $f$ (seen as a vectorial bent function from $V_{n}$ to $V_{1}$) is vectorial dual-bent with $(cf)^{*}=\sigma(c)f^{*}, \sigma(c)=c^{1-t}, c \in \mathbb{F}_{p}^{*}$, where $(l-1)(t-1)\equiv 1 \ mod \ (p-1)$. It is easy to see that $\sigma$ satisfies the condition of Corollary 2 and thus $D_{S}^{f}, D_{N}^{f}$ are partial difference sets and any union of sets from $\{D_{0}^{f}, D_{S}^{f}, D_{N}^{f}\}$ is a partial difference set, which is the result given in Proposition 1 (ii).
\end{remark}

\begin{remark}\label{Remark5}
Let $m$ be a positive even integer. By the analysis on nonsingular quadratic forms in Section 2, we can see that any nonsingular diagonal quadratic form $G: \mathbb{F}_{p^s}^{m}\rightarrow \mathbb{F}_{p^s}$ defined as $G(x_{1}, \dots, x_{m})=a_{1}x_{1}^{2}+\dots+a_{m}x_{m}^{2}$ satisfies the corresponding condition of Corollary 2, where $a_{i} \in \mathbb{F}_{p^s}^{*}, 1\leq i \leq m$. Then by Corollary 2, $D_{S}^{G}$ is a partial difference set. As we illustrate in Section 2, any nonsingular quadratic form $F: \mathbb{F}_{p^s}^{m}\rightarrow \mathbb{F}_{p^s}$ is linear equivalent to some nonsingular diagonal quadratic form $G$, thus $D_{S}^{F}$ is a partial difference set, which is the result given in Proposition 2 (iv). Furthermore, by Corollary 2, any union of sets from $\{D_{0}^{F}, D_{S}^{F}, D_{N}^{F}\}$ is a partial difference set.
\end{remark}

By Theorem 1 and its proof, it is easy to see that the (non)-quadratic vectorial dual-bent functions defined by (8) for which $2s \mid n$ and $\eta_{n}(\alpha_{i}), 1\leq i \leq 3$ are the same satisfy the corresponding condition of Corollary 2, thus we obtain the following explicit construction of partial difference sets.

\begin{corollary}\label{Corollary3}
Let $s, m, n$ be positive integers with $s \mid m$ and $2s \mid n$, and let $S$ (respectively, $N$) denote the set of all squares (respectively, non-squares) in $\mathbb{F}_{p^s}^{*}$. Define $F_{i} \ (i \in \mathbb{F}_{p^s}): \mathbb{F}_{p^n}\rightarrow \mathbb{F}_{p^s}$ as $F_{0}(x)=Tr_{s}^{n}(\alpha_{1}x^{2})$, $F_{i}(x)=Tr_{s}^{n}(\alpha_{2}x^{2})$ if $i$ is a square in $\mathbb{F}_{p^s}^{*}$, $F_{i}(x)=Tr_{s}^{n}(\alpha_{3}x^{2})$ if $i$ is a non-square in $\mathbb{F}_{p^s}^{*}$, where $\alpha_{j} \in \mathbb{F}_{p^n}^{*}, 1\leq j\leq 3$ are all squares or non-squares. Define $G: \mathbb{F}_{p^m} \times \mathbb{F}_{p^m}\rightarrow \mathbb{F}_{p^s}$ as $G(y_{1}, y_{2})=Tr_{s}^{m}(\beta y_{1}L(y_{2}))$, where $\beta \in \mathbb{F}_{p^m}^{*}$ and $L(x)=\sum a_{i}x^{q^{i}} \ (q=p^s)$ is a $q$-polynomial over $\mathbb{F}_{p^m}$ inducing a permutation of $\mathbb{F}_{p^m}$. Let $T: \mathbb{F}_{p^n} \times \mathbb{F}_{p^m} \times \mathbb{F}_{p^m}\rightarrow \mathbb{F}_{p^s}$ be defined as $T(x, y_{1}, y_{2})=F_{Tr_{s}^{m}(\gamma y_{2}^{2})}(x)+G(y_{1}, y_{2})$, where $\gamma \in \mathbb{F}_{p^s}^{*}$, and let $\varepsilon=(-1)^{n-1}\epsilon^{n}\eta_{n}(\alpha_{1})$, where $\epsilon=1$ if $p\equiv 1 \ (mod \ 4)$ and $\epsilon=\sqrt{-1}$ if $p\equiv 3 \ (mod \ 4)$, $\eta_{n}$ is the multiplicative quadratic character of $\mathbb{F}_{p^n}$. Then $D_{S}^{T}=\{(x, y_{1}, y_{2}) \in \mathbb{F}_{p^n}\times \mathbb{F}_{p^m} \times \mathbb{F}_{p^m}: T(x, y_{1}, y_{2}) \in S\}$ and $D_{N}^{T}=\{(x, y_{1}, y_{2}) \in \mathbb{F}_{p^n}\times \mathbb{F}_{p^m} \times \mathbb{F}_{p^m}: T(x, y_{1}, y_{2}) \in N\}$ are both $(p^{n+2m}, k, \lambda, \mu)$ partial difference sets, where
\begin{equation*}
\begin{split}
& k=\frac{p^s-1}{2}p^{n+2m-s}-\varepsilon p^{\frac{n}{2}+m-s}\frac{p^s-1}{2},\\
& \lambda=(\frac{p^s-1}{2})^{2}p^{n+2m-2s}-\varepsilon p^{\frac{n}{2}+m-s}\frac{p^s-3}{2},\\
& \mu=(\frac{p^s-1}{2})^{2}p^{n+2m-2s}-\varepsilon p^{\frac{n}{2}+m-s}\frac{p^s-1}{2}.
\end{split}
\end{equation*}
Furthermore, any union of sets from $\{D_{0}^{T}, D_{S}^{T}, D_{N}^{T}\}$, denoted by $A$, is a $(p^{n+2m}, $ $k, \lambda, \mu)$ partial difference set, where
\begin{equation*}
\begin{split}
& k=(m_{1}\frac{p^s-1}{2}+m_{0})p^{n+2m-s}+\varepsilon p^{\frac{n}{2}+m-s}(m_{0}p^s-(m_{1}\frac{p^s-1}{2}+m_{0}))-m_{0},\\
\end{split}
\end{equation*}
\begin{equation*}
\begin{split}
& \lambda=(m_{1}\frac{p^s-1}{2}+m_{0})^{2}p^{n+2m-2s}+\varepsilon p^{\frac{n}{2}+m-s}(p^s+(2m_{0}-3)m_{1}\frac{p^s-1}{2}-m_{0})\\
& \ \ \ \ \ -2m_{0},\\
& \mu=(m_{1}\frac{p^s-1}{2}+m_{0})^{2}p^{n+2m-2s}+\varepsilon p^{\frac{n}{2}+m-s}((2m_{0}-1)m_{1}\frac{p^s-1}{2}+m_{0}),
\end{split}
\end{equation*}
$0\leq m_{1}\leq 2$ is the number of distinct $D_{S}^{T}, D_{N}^{T}$ in $A$, $m_{0}=1$ if $ D_{0}^{T} \subseteq A$ and $m_{0}=0$ if $D_{0}^{T} \not \subseteq A$.
\end{corollary}

We give an example by using Corollary 3.

\begin{example}\label{Example2}
Let $F_{i} \ (i \in \mathbb{F}_{5^{2}}): \mathbb{F}_{5^8} \rightarrow \mathbb{F}_{5^2}$ be defined by $F_{0}(x)=Tr_{2}^{8}(x^{2})$, $F_{i}(x)=Tr_{2}^{8}(w^{2}x^{2})$ if $i \in \mathbb{F}_{5^2}^{*}$, where $w$ is a primitive element of $\mathbb{F}_{5^8}$. Define $T: \mathbb{F}_{5^8} \times \mathbb{F}_{5^4} \times \mathbb{F}_{5^4}\rightarrow \mathbb{F}_{5^2}$ as $T(x, y_{1}, y_{2})=F_{Tr_{2}^{4}(y_{2}^{2})}(x)+Tr_{2}^{4}(y_{1}y_{2})=(Tr_{2}^{4}(y_{2}^{2}))^{24}Tr_{2}^{8}((w^{2}-1)x^{2})+Tr_{2}^{8}(x^{2})+Tr_{2}^{4}(y_{1}y_{2})$.
By Corollary 3, $D_{S}^{T}, D_{N}^{T}$ are both $(152587890625, 73242375000, 35156421875, 35156437500)$ partial difference sets, and $D_{0}^{T}\cup D_{S}^{T}$ is a $(152587890625, 79345515624, 41259578123, $ $41259562500)$ partial difference set.
\end{example}

It is easy to see that the (non)-quadratic vectorial dual-bent functions defined by (3) satisfy the corresponding condition of Theorem 3, thus we obtain the following explicit construction of partial difference sets.

\begin{corollary}\label{Corollary4}
Let $s$ be a divisor of $m$ and $F: \mathbb{F}_{p^m} \times \mathbb{F}_{p^m}\rightarrow \mathbb{F}_{p^s}$ be defined as
 $F(x, y)=Tr_{s}^{m}(axy^{e})$,
where $a \in \mathbb{F}_{p^m}^{*}, gcd(e, p^m-1)=1$. Let $l$ be a positive integer with $gcd(l, p^s-1) \mid (1+e)$ and $H_{l}=\{x^{l}: x \in \mathbb{F}_{p^s}^{*}\}$. Then for any $\beta \in \mathbb{F}_{p^s}^{*}$, $D_{\beta H_{l}}^{F}=\{(x, y) \in \mathbb{F}_{p^m} \times \mathbb{F}_{p^m}: F(x, y) \in \beta H_{l}\}$ is a $(p^{2m}, k, \lambda, \mu)$ partial difference set, where
\begin{equation*}
  \begin{split}
      & k=\frac{p^s-1}{gcd(l, p^s-1)}p^{2m-s}-p^{m-s}\frac{p^s-1}{gcd(l, p^s-1)},\\
      & \lambda=(\frac{p^s-1}{gcd(l, p^s-1)})^{2}p^{2m-2s}+p^{m-s}(p^s-3\frac{p^s-1}{gcd(l, p^s-1)}),\\
      & \mu=(\frac{p^s-1}{gcd(l, p^s-1)})^{2}p^{2m-2s}-p^{m-s}\frac{p^s-1}{gcd(l, p^s-1)}.
  \end{split}
\end{equation*}
Furthermore, any union of sets from $\{D_{0}^{F}, D_{\beta H_{l}}^{F}, \beta \in \mathbb{F}_{p^s}^{*}\}$, denoted by $A$, is a $(p^{2m}, k, \lambda, \mu)$ partial difference set, where
\begin{equation*}
  \begin{split}
      & k=(m_{1}\frac{p^s-1}{gcd(l, p^s-1)}+m_{0})p^{2m-s}+p^{m-s}(m_{0}p^s-(m_{1}\frac{p^s-1}{gcd(l, p^s-1)}+m_{0}))\\
      & \ \ \ \ \ -m_{0},\\
      & \lambda=(m_{1}\frac{p^s-1}{gcd(l, p^s-1)}+m_{0})^{2}p^{2m-2s}+p^{m-s}(p^s+(2m_{0}-3)m_{1}\frac{p^s-1}{gcd(l, p^s-1)}\\
      & \ \ \ \ \ -m_{0})-2m_{0},\\
      & \mu=(m_{1}\frac{p^s-1}{gcd(l, p^s-1)}+m_{0})^{2}p^{2m-2s}+p^{m-s}((2m_{0}-1)m_{1}\frac{p^s-1}{gcd(l, p^s-1)}\\
      & \ \ \ \ \ +m_{0}),
  \end{split}
\end{equation*}
$0\leq m_{1}\leq gcd(l, p^s-1)$ is the number of distinct $D_{\beta H_{l}}^{F}$ in $A$, $m_{0}=1$ if $ D_{0}^{F} \subseteq A$ and $m_{0}=0$ if $D_{0}^{F} \not \subseteq A$.
\end{corollary}

\begin{remark}\label{Remark6}
Let $F: \mathbb{F}_{p^m} \times \mathbb{F}_{p^m}\rightarrow \mathbb{F}_{p^s}$ be the Maiorana-McFarland vectorial bent function $F(x, y)=Tr_{s}^{m}(xy)$. Then by Corollary 4, the preimage set of all squares (respectively, non-squares) of $\mathbb{F}_{p^s}^{*}$ for $F$ is a partial difference set. Further, for $G(x, y)=Tr_{s}^{m}(xL(y))$, where $L$ is an arbitrary linear permutation over $\mathbb{F}_{p^m}$, since $F$ and $G$ are linear equivalent, we have that the preimage set of all squares (respectively, non-squares) of $\mathbb{F}_{p^s}^{*}$ for $G$ is a partial difference set, which is the result given in Proposition 2 (iii).
\end{remark}

We give an example by using Corollary 4.

\begin{example}\label{Example3}
Let $F: \mathbb{F}_{7^4} \times \mathbb{F}_{7^4}\rightarrow \mathbb{F}_{7^2}$ be defined as $F(x_{1}, x_{2})=Tr_{2}^{4}(x_{1}x_{2}^{11})$. Let $l=3$, then $gcd(l, 7^2-1)=3 \mid (1+11)$. By Corollary 4, $D_{\beta H_{3}}^{F}$ is a $(5764801, 1881600, 614705, 613872)$ partial difference set for any $\beta \in \mathbb{F}_{7^2}^{*}$, and $D_{0}^{F} \cup D_{H_{3}}^{F}$ is a $(5764801, 2001600, 695455, 694722)$ partial difference set.
\end{example}

We can see that if $F: V_{n}\rightarrow \mathbb{F}_{p^s}$ is vectorial dual-bent satisfying the corresponding condition of Corollary 1 and $(F_{c})^{*}=(F^{*})_{\sigma(c)}, \sigma(c)=c^{-1}, c \in \mathbb{F}_{p^s}^{*}$ (where $F^{*}$ is a vectorial dual of $F$), then $\sigma^{-1}(c)H_{l}=c H_{l}$ if and only if $H_{l}=H_{2}$ or $H_{l}=\mathbb{F}_{p^s}^{*}$ and the partial different sets constructed by Theorem 3 are the unions of sets from $\{D_{0}^{F}, D_{H_{2}}^{F}, D_{\beta H_{2}}^{F}\}$ (where $\beta \in \mathbb{F}_{p^s}^{*}$ is a non-square) or $\{D_{0}^{F}, D_{\mathbb{F}_{p^s}^{*}}^{F}\}$. For such a vectorial dual-bent function, we want to find other subgroup $H_{t}$ of $\mathbb{F}_{p^s}^{*}$ different from $H_{2}$ and $\mathbb{F}_{p^s}^{*}$ such that $D_{H_{t}}^{F}$ is a partial difference set. The following theorem shows that if $s=2jr$ and $t \mid (p^{j}+1)$, where $j$ is the smallest such positive integer, then any union of sets from $\{D_{0}^{F}, D_{\beta H_{t}}^{F}, \beta \in \mathbb{F}_{p^s}^{*}\}$ is a partial difference set. Note that in this case, $H_{t}$ can be different from $H_{2}$ and $\mathbb{F}_{p^s}^{*}$.

\begin{theorem}\label{Theorem4}
Let $p$ be an odd prime, $n, s, t$ be positive integers which satisfy that $n$ is even, $s=2jr$ for some positive integers $j, r$, $t\geq 2$ and $t \mid (p^{j}+1)$, where $j$ is the smallest such positive integer. Let $F: V_{n}\rightarrow \mathbb{F}_{p^s}$ be a vectorial dual-bent function which satisfies that $F(0)=0, F(-x)=F(x)$ and all component functions are regular (that is, $\varepsilon_{F_{c}}=1$ for all $c \in \mathbb{F}_{p^s}^{*}$) or weakly regular but not regular (that is, $\varepsilon_{F_{c}}=-1$ for all $c \in \mathbb{F}_{p^s}^{*}$). Suppose $\varepsilon_{F_{c}}=\varepsilon, c \in \mathbb{F}_{p^s}^{*}$. Let $F^{*}$ be a vectorial dual of $F$ and $(F_{c})^{*}=(F^{*})_{\sigma(c)}, c \in \mathbb{F}_{p^s}^{*}$ for a permutation $\sigma$ over $\mathbb{F}_{p^s}^{*}$. If $\sigma$ satisfies that for any $\beta \in \mathbb{F}_{p^s}^{*}$, there exists $\beta' \in \mathbb{F}_{p^s}^{*}$ such that $\sigma(\beta H_{t})=\beta'H_{t}$, where $H_{t}=\{x^{t}: x \in \mathbb{F}_{p^s}^{*}\}$, then for any $\beta \in \mathbb{F}_{p^s}^{*}$, $D_{\beta H_{t}}^{F}=\{x \in V_{n}^{*}: F(x) \in \beta H_{t}\}$ is a $(p^n, k, \lambda, \mu)$ partial difference set, where
\begin{equation*}
  \begin{split}
      & k=\frac{p^s-1}{t}p^{n-s}-\varepsilon p^{\frac{n}{2}-s}\frac{p^s-1}{t},\\
      & \lambda=(\frac{p^s-1}{t})^{2}p^{n-2s}+\varepsilon p^{\frac{n}{2}-s}(p^s-3\frac{p^s-1}{t}),\\
      & \mu=(\frac{p^s-1}{t})^{2}p^{n-2s}-\varepsilon p^{\frac{n}{2}-s}\frac{p^s-1}{t}.
  \end{split}
\end{equation*}
Furthermore, any union of sets from $\{D_{0}^{F}, D_{\beta H_{t}}^{F}, \beta \in \mathbb{F}_{p^s}^{*}\}$, denoted by $A$, is a $(p^n, k, \lambda, \mu)$ partial difference set, where
\begin{equation*}
  \begin{split}
      & k=(m_{1}\frac{p^s-1}{t}+m_{0})p^{n-s}+\varepsilon p^{\frac{n}{2}-s}(m_{0}p^s-(m_{1}\frac{p^s-1}{t}+m_{0}))-m_{0},\\
      & \lambda=(m_{1}\frac{p^s-1}{t}+m_{0})^{2}p^{n-2s}+\varepsilon p^{\frac{n}{2}-s}(p^s+(2m_{0}-3)m_{1}\frac{p^s-1}{t}-m_{0})-2m_{0},\\
      & \mu=(m_{1}\frac{p^s-1}{t}+m_{0})^{2}p^{n-2s}+\varepsilon p^{\frac{n}{2}-s}((2m_{0}-1)m_{1}\frac{p^s-1}{t}+m_{0}),
  \end{split}
\end{equation*}
$0\leq m_{1}\leq t$ is the number of distinct $D_{\beta H_{t}}^{F}$ in $A$, $m_{0}=1$ if $ D_{0}^{F} \subseteq A$ and $m_{0}=0$ if $D_{0}^{F} \not \subseteq A$.
\end{theorem}
\begin{proof}
First by $s=2jr$ and $t \mid (p^{j}+1)$, we have $t \mid (p^s-1)$ and thus $gcd(t, p^s-1)=t$. For any $\beta \in \mathbb{F}_{p^s}^{*}$, $u \in V_{n}^{*}$, as given in the proof of Theorem 3, we have
\begin{equation*}
\begin{split}
    \chi_{u}(D_{\beta H_{t}}^{F})
                & =\varepsilon p^{\frac{n}{2}-s}\sum_{c \in \mathbb{F}_{p^s}^{*}}\zeta_{p}^{Tr_{1}^{s}(\sigma(c)F^{*}(-u))}\sum_{i \in H_{t}}\zeta_{p}^{Tr_{1}^{s}(-c \beta i)}\\
\end{split}
\end{equation*}
\begin{equation}\label{20}
\begin{split}
               & =\varepsilon p^{\frac{n}{2}-s}\sum_{c \in \mathbb{F}_{p^s}^{*}}\zeta_{p}^{Tr_{1}^{s}(\sigma(c)F^{*}(-u))}\eta_{-c\beta}^{(t, p^s)},
    \end{split}
\end{equation}
where for any $a \in \mathbb{F}_{p^s}^{*}$, $\eta_{a}^{(t, p^s)}\triangleq \sum_{x \in H_{t}}\zeta_{p}^{Tr_{1}^{s}(a x)}$.

In the case that $s=2jr$, $t\geq 2$ and $t \mid (p^{j}+1)$, where $j$ is the smallest such positive integer, by the well-known Gaussian periods in the semiprimitive case (see for example Theorem 1.18 of \cite{Ding}), we have
\begin{equation}\label{21}
  \eta_{a}^{(t, p^s)}=\left\{
  \begin{split}
  &\delta_{w^{\frac{t}{2}}H_{t}}(a)p^{\frac{s}{2}}-\frac{p^\frac{s}{2}+1}{t}, \ \ \ \ \ \ \ \ \ \ \ \ \ \ \ \ \text{if} \ r \ \text{and} \ \frac{p^j+1}{t} \ \text{are both odd},\\
  &\delta_{H_{t}}(a)(-1)^{r+1}p^{\frac{s}{2}}+\frac{(-1)^{r}p^{\frac{s}{2}}-1}{t}, \ \ \text{otherwise},
  \end{split}
  \right.
\end{equation}
where $w$ is an arbitrary fixed primitive element of $\mathbb{F}_{p^s}$ (Note that when $r$ and $\frac{p^j+1}{t}$ are both odd, $w^{\frac{t}{2}}H_{t}$ is independent of $w$) and for any set $A$, $\delta_{A}: A\rightarrow \{0, 1\}$ is defined by
\begin{equation*}
  \delta_{A}(a)=\left\{
  \begin{split}
  & 1, \ \text{if} \ a \in A,\\
  & 0, \ \text{if} \ a \notin A.
  \end{split}\right.
\end{equation*}
(i) If $r$ and $\frac{p^j+1}{t}$ are both odd, then by (20) and (21), we have
\begin{equation*}
  \begin{split}
    \chi_{u}(D_{\beta H_{t}}^{F})
                & =\varepsilon p^{\frac{n}{2}-s}\sum_{c \in \mathbb{F}_{p^s}^{*}}\zeta_{p}^{Tr_{1}^{s}(\sigma(c)F^{*}(-u))}(-\frac{p^{\frac{s}{2}}+1}{t})\\
                & \ \ \ +\varepsilon p^{\frac{n-s}{2}}\sum_{-c\beta \in w^{\frac{t}{2}}H_{t}}\zeta_{p}^{Tr_{1}^{s}(\sigma(c)F^{*}(-u))}\\
                &=\varepsilon p^{\frac{n}{2}-s}\sum_{c \in \mathbb{F}_{p^s}^{*}}\zeta_{p}^{Tr_{1}^{s}(cF^{*}(-u))}(-\frac{p^{\frac{s}{2}}+1}{t})\\
                & \ \ \ +\varepsilon p^{\frac{n-s}{2}}\sum_{c\in H_{t}}\zeta_{p}^{Tr_{1}^{s}(\sigma(-\beta^{-1}w^{\frac{t}{2}}c)F^{*}(-u))}\\
                &=\varepsilon p^{\frac{n}{2}-s}\frac{p^{\frac{s}{2}}+1}{t}(-\delta_{\{0\}}(F^{*}(-u))p^s+1)\\
                & \ \ \ +\varepsilon p^{\frac{n-s}{2}}\sum_{c\in H_{t}}\zeta_{p}^{Tr_{1}^{s}(\sigma(-\beta^{-1}w^{\frac{t}{2}}c)F^{*}(-u))}.\\
    \end{split}
\end{equation*}
If for any $\gamma \in \mathbb{F}_{p^s}^{*}$, there exists $\gamma' \in \mathbb{F}_{p^s}^{*}$ such that $\sigma(\gamma H_{t})=\gamma'H_{t}$, then $\sigma(-\beta^{-1}w^{\frac{t}{2}}H_{t})=w^{i}H_{t}$ for some $0\leq i \leq t-1$ and we have
\begin{equation*}
  \begin{split}
    \chi_{u}(D_{\beta H_{t}}^{F})
                &=\varepsilon p^{\frac{n}{2}-s}\frac{p^{\frac{s}{2}}+1}{t}(-\delta_{\{0\}}(F^{*}(-u))p^s+1)\\
                & \ \ \ +\varepsilon p^{\frac{n-s}{2}}\sum_{c\in H_{t}}\zeta_{p}^{Tr_{1}^{s}(cw^{i}F^{*}(-u))}.\\
    \end{split}
\end{equation*}
When $F^{*}(-u)=0$, we have
\begin{equation}\label{22}
  \begin{split}
    \chi_{u}(D_{\beta H_{t}}^{F})
                &=\varepsilon p^{\frac{n}{2}-s}\frac{p^{\frac{s}{2}}+1}{t}(-p^s+1)+\varepsilon p^{\frac{n-s}{2}}|H_{t}|\\
                &=\varepsilon p^{\frac{n}{2}-s}\frac{p^{\frac{s}{2}}+1}{t}(-p^s+1)+\varepsilon p^{\frac{n-s}{2}}\frac{p^s-1}{t}\\
                &=-\varepsilon p^{\frac{n}{2}-s}\frac{p^s-1}{t}.
    \end{split}
\end{equation}
When $F^{*}(-u)\neq 0$, we have
\begin{equation}\label{23}
  \begin{split}
    \chi_{u}(D_{\beta H_{t}}^{F})
                &=\varepsilon p^{\frac{n}{2}-s}\frac{p^{\frac{s}{2}}+1}{t}+\varepsilon p^{\frac{n-s}{2}}\eta_{w^{i}F^{*}(-u)}^{(t, p^s)}\\
                &=\varepsilon p^{\frac{n}{2}-s}\frac{p^{\frac{s}{2}}+1}{t}+\varepsilon p^{\frac{n-s}{2}}(\delta_{w^{\frac{t}{2}}H_{t}}(w^{i}F^{*}(-u))
                p^{\frac{s}{2}}-\frac{p^{\frac{s}{2}}+1}{t}).\\
    \end{split}
\end{equation}
By (22) and (23), we have
\begin{equation}\label{24}
  \chi_{u}(D_{\beta H_{t}}^{F})=\left\{
  \begin{split}
  -\varepsilon p^{\frac{n}{2}-s}\frac{p^s-1}{t}, & \ \ \text{if} \ w^{i}F^{*}(-u) \notin w^{\frac{t}{2}}H_{t},\\
  \varepsilon p^{\frac{n}{2}}-\varepsilon p^{\frac{n}{2}-s}\frac{p^s-1}{t}, & \ \ \text{if} \ w^{i}F^{*}(-u) \in w^{\frac{t}{2}}H_{t}.
  \end{split}
  \right.
\end{equation}
Then by Corollary 1 and Lemma 1, $D_{\beta H_{t}}^{F}=\{x \in V_{n}^{*}: F(x) \in \beta H_{t}\}$ is a $(p^n, k, \lambda, \mu)$ partial difference set, where
\begin{equation*}
  \begin{split}
      & k=\frac{p^s-1}{t}p^{n-s}-\varepsilon p^{\frac{n}{2}-s}\frac{p^s-1}{t},\\
      & \lambda=(\frac{p^s-1}{t})^{2}p^{n-2s}+\varepsilon p^{\frac{n}{2}-s}(p^s-3\frac{p^s-1}{t}),\\
      & \mu=(\frac{p^s-1}{t})^{2}p^{n-2s}-\varepsilon p^{\frac{n}{2}-s}\frac{p^s-1}{t}.
  \end{split}
\end{equation*}
Note that by the assumption on $\sigma$, for any $0\leq i_{1}< \dots < i_{m_{1}}\leq [\mathbb{F}_{p^s}^{*}: H_{t}]-1=t-1$, there exist distinct $0\leq z_{1}, \dots, z_{m_{1}}\leq t-1$ such that $\sigma(-w^{-i_{j}+\frac{t}{2}}H_{t})=w^{z_{j}}H_{t}$ for any $1\leq j \leq m_{1}$. Then for any $u \in V_{n}^{*}$, by (24) we have
\begin{equation}\label{25}
\begin{split}
  &\sum_{j=1}^{m_{1}}\chi_{u}(D_{w^{i_{j}} H_{t}}^{F})\\
  &=\left\{
  \begin{split}
  &-\varepsilon m_{1}p^{\frac{n}{2}-s}\frac{p^s-1}{t}, \ \ \ \ \ \ \ \text{if} \ F^{*}(-u) \notin \cup_{j=1}^{m_{1}}w^{-z_{j}+\frac{t}{2}}H_{t},\\
  &\varepsilon p^{\frac{n}{2}}-\varepsilon m_{1}p^{\frac{n}{2}-s}\frac{p^s-1}{t}, \ \ \text{if} \ F^{*}(-u) \in w^{-z_{j_{0}}+\frac{t}{2}}H_{t} \ \text{for some} \ 1\leq j_{0}\leq m_{1}.\\
  \end{split}\right.
  \end{split}
\end{equation}
By (18) and (25), we have
\begin{equation}\label{26}
\begin{split}
  &\sum_{j=1}^{m_{1}}\chi_{u}(D_{w^{i_{j}} H_{t}}^{F})+\chi_{u}(D_{0}^{F})\\
  &=\left\{
  \begin{split}
  & \varepsilon p^{\frac{n}{2}-s}\frac{(p^s-1)(t-m_{1})}{t}-1,  \ \text{if} \ F^{*}(-u)=0 \ \text{or} \ F^{*}(-u) \in w^{-z_{j_{0}}+\frac{t}{2}}H_{t}\\
  & \ \ \ \ \ \ \ \ \ \ \ \ \ \ \ \ \ \ \ \ \ \ \ \ \ \ \ \ \ \ \ \ \ \ \ \ \ \ \text{for some} \ 1\leq j_{0}\leq m_{1},\\
  & \varepsilon p^{\frac{n}{2}-s}(-m_{1}\frac{p^s-1}{t}-1)-1, \ \text{if} \ F^{*}(-u)\neq 0 \\
  & \ \ \ \ \ \ \ \ \ \ \ \ \ \ \ \ \ \ \ \ \ \ \ \ \ \ \ \ \ \ \ \ \ \ \ \ \ \ \text{and} \ F^{*}(-u) \notin \cup_{j=1}^{m_{1}}w^{-z_{j}+\frac{t}{2}}H_{t}. \\
  \end{split}\right.
  \end{split}
\end{equation}
By Corollary 1, Lemma 1 and (18), (25), (26), any union of sets from $\{D_{0}^{F},$ $ D_{\beta H_{t}}^{F}, \beta \in \mathbb{F}_{p^s}^{*}\}$, denoted by $A$, is a $(p^n, k, \lambda, \mu)$ partial difference set, where
\begin{equation*}
  \begin{split}
      & k=(m_{1}\frac{p^s-1}{t}+m_{0})p^{n-s}+\varepsilon p^{\frac{n}{2}-s}(m_{0}p^s-(m_{1}\frac{p^s-1}{t}+m_{0}))-m_{0},\\
      & \lambda=(m_{1}\frac{p^s-1}{t}+m_{0})^{2}p^{n-2s}+\varepsilon p^{\frac{n}{2}-s}(p^s+(2m_{0}-3)m_{1}\frac{p^s-1}{t}-m_{0})-2m_{0},\\
      & \mu=(m_{1}\frac{p^s-1}{t}+m_{0})^{2}p^{n-2s}+\varepsilon p^{\frac{n}{2}-s}((2m_{0}-1)m_{1}\frac{p^s-1}{t}+m_{0}),
  \end{split}
\end{equation*}
$0\leq m_{1}\leq t$ is the number of distinct $D_{\beta H_{t}}^{F}$ in $A$, $m_{0}=1$ if $ D_{0}^{F} \subseteq A$ and $m_{0}=0$ if $D_{0}^{F} \not \subseteq A$.

(ii) If $r$ or $\frac{p^j+1}{t}$ is even, then by (20) and (21), we have
\begin{equation*}
  \begin{split}
    \chi_{u}(D_{\beta H_{t}}^{F})
                & =\varepsilon p^{\frac{n}{2}-s}\sum_{c \in \mathbb{F}_{p^s}^{*}}\zeta_{p}^{Tr_{1}^{s}(\sigma(c)F^{*}(-u))}\frac{(-1)^{r}p^{\frac{s}{2}}-1}{t}\\
                & \ \ \ +\varepsilon (-1)^{r+1}p^{\frac{n-s}{2}}\sum_{-c\beta \in H_{t}}\zeta_{p}^{Tr_{1}^{s}(\sigma(c)F^{*}(-u))}\\
                &=\varepsilon p^{\frac{n}{2}-s}\sum_{c \in \mathbb{F}_{p^s}^{*}}\zeta_{p}^{Tr_{1}^{s}(cF^{*}(-u))}\frac{(-1)^{r}p^{\frac{s}{2}}-1}{t}\\
                &\ \ +\varepsilon (-1)^{r+1}p^{\frac{n-s}{2}}\sum_{c\in H_{t}}\zeta_{p}^{Tr_{1}^{s}(\sigma(-\beta^{-1}c)F^{*}(-u))}\\
                &=\varepsilon p^{\frac{n}{2}-s}\frac{(-1)^{r}p^{\frac{s}{2}}-1}{t}(\delta_{\{0\}}(F^{*}(-u))p^s-1)\\
                &\ \ +\varepsilon (-1)^{r+1}p^{\frac{n-s}{2}}\sum_{c\in H_{t}}\zeta_{p}^{Tr_{1}^{s}(\sigma(-\beta^{-1}c)F^{*}(-u))}.\\
    \end{split}
\end{equation*}
If for any $\gamma \in \mathbb{F}_{p^s}^{*}$, there exists $\gamma' \in \mathbb{F}_{p^s}^{*}$ such that $\sigma(\gamma H_{t})=\gamma'H_{t}$, then $\sigma(-\beta^{-1}H_{t})=w^{i}H_{t}$ for some $0\leq i \leq t-1$ and we have
\begin{equation*}
  \begin{split}
    \chi_{u}(D_{\beta H_{t}}^{F})
                &=\varepsilon p^{\frac{n}{2}-s}\frac{(-1)^{r}p^{\frac{s}{2}}-1}{t}(\delta_{\{0\}}(F^{*}(-u))p^s-1)\\
                & \ \ \ +\varepsilon(-1)^{r+1}p^{\frac{n-s}{2}}\sum_{c\in H_{t}}\zeta_{p}^{Tr_{1}^{s}(cw^{i}F^{*}(-u))}.\\
    \end{split}
\end{equation*}
When $F^{*}(-u)=0$, we have
\begin{equation}\label{27}
  \begin{split}
    \chi_{u}(D_{\beta H_{t}}^{F})
                &=\varepsilon p^{\frac{n}{2}-s}\frac{(-1)^{r}p^{\frac{s}{2}}-1}{t}(p^s-1)+\varepsilon(-1)^{r+1}p^{\frac{n-s}{2}}|H_{t}|\\
                &=\varepsilon p^{\frac{n}{2}-s}\frac{(-1)^{r}p^{\frac{s}{2}}-1}{t}(p^s-1)+\varepsilon(-1)^{r+1}p^{\frac{n-s}{2}}\frac{p^s-1}{t}\\
                &=-\varepsilon p^{\frac{n}{2}-s}\frac{p^s-1}{t}.\\
    \end{split}
\end{equation}
When $F^{*}(-u)\neq 0$, we have
\begin{equation}\label{28}
  \begin{split}
    \chi_{u}(D_{\beta H_{t}}^{F})
                &=-\varepsilon p^{\frac{n}{2}-s}\frac{(-1)^{r}p^{\frac{s}{2}}-1}{t}+\varepsilon(-1)^{r+1}p^{\frac{n-s}{2}}\eta_{w^{i}F^{*}(-u)}^{(t, p^s)}\\
                &=-\varepsilon p^{\frac{n}{2}-s}\frac{(-1)^{r}p^{\frac{s}{2}}-1}{t}\\
                & \ \ \ +\varepsilon p^{\frac{n-s}{2}}(\delta_{H_{t}}(w^{i}F^{*}(-u))
                p^{\frac{s}{2}}+(-1)^{r+1}\frac{(-1)^{r}p^{\frac{s}{2}}-1}{t}).\\
    \end{split}
\end{equation}
By (27) and (28), we have
\begin{equation}\label{29}
  \chi_{u}(D_{\beta H_{t}}^{F})=\left\{
  \begin{split}
  -\varepsilon p^{\frac{n}{2}-s}\frac{p^s-1}{t}, & \ \ \text{if} \ w^{i}F^{*}(-u) \notin H_{t},\\
  \varepsilon p^{\frac{n}{2}}-\varepsilon p^{\frac{n}{2}-s}\frac{p^s-1}{t},  & \ \ \text{if} \ w^{i}F^{*}(-u) \in H_{t}.
  \end{split}
  \right.
\end{equation}
Then by Corollary 1 and Lemma 1, $D_{\beta H_{t}}^{F}=\{x \in V_{n}^{*}: F(x) \in \beta H_{t}\}$ is a $(p^n, k, \lambda, \mu)$ partial difference set, where
\begin{equation*}
  \begin{split}
      & k=\frac{p^s-1}{t}p^{n-s}-\varepsilon p^{\frac{n}{2}-s}\frac{p^s-1}{t},\\
      & \lambda=(\frac{p^s-1}{t})^{2}p^{n-2s}+\varepsilon p^{\frac{n}{2}-s}(p^s-3\frac{p^s-1}{t}),\\
      & \mu=(\frac{p^s-1}{t})^{2}p^{n-2s}-\varepsilon p^{\frac{n}{2}-s}\frac{p^s-1}{t}.
  \end{split}
\end{equation*}
By the assumption on $\sigma$, for any $0\leq i_{1}< \dots < i_{m_{1}}\leq[\mathbb{F}_{p^s}^{*}: H_{t}]-1=t-1$, there exist distinct $0\leq z_{1}, \dots, z_{m_{1}}\leq t-1$ such that $\sigma(-w^{-i_{j}}H_{t})=w^{z_{j}}H_{t}$ for any $1\leq j \leq m_{1}$. Then for any $u \in V_{n}^{*}$, by (29) we have
\begin{equation}\label{30}
\begin{split}
  &\sum_{j=1}^{m_{1}}\chi_{u}(D_{w^{i_{j}} H_{t}}^{F})\\
  &=\left\{
  \begin{split}
  & -\varepsilon m_{1}p^{\frac{n}{2}-s}\frac{p^s-1}{t}, \ \ \ \ \ \ \text{if} \ F^{*}(-u) \notin \cup_{j=1}^{m_{1}}w^{-z_{j}}H_{t},\\
  &\varepsilon p^{\frac{n}{2}}-\varepsilon m_{1}p^{\frac{n}{2}-s}\frac{p^s-1}{t}, \ \text{if} \ F^{*}(-u) \in w^{-z_{j_{0}}}H_{t} \ \  \text{for some} \ 1\leq j_{0}\leq m_{1}.
  \end{split}
  \right.
  \end{split}
\end{equation}
By (18) and (30),
\begin{equation}\label{31}
\begin{split}
  &\sum_{j=1}^{m_{1}}\chi_{u}(D_{w^{i_{j}} H_{t}}^{F})+\chi_{u}(D_{0}^{F})\\
  &=\left\{
  \begin{split}
  & \varepsilon p^{\frac{n}{2}-s}\frac{(p^s-1)(t-m_{1})}{t}-1, \ \text{if} \ F^{*}(-u)=0 \ \text{or} \ F^{*}(-u) \in w^{-z_{j_{0}}}H_{t} \\
  & \ \ \ \ \ \ \ \ \ \ \ \ \ \ \ \ \ \ \ \ \ \ \ \ \ \ \ \ \ \ \ \ \ \ \ \ \ \ \text{for some} \ 1\leq j_{0}\leq m_{1},\\
  &  \varepsilon p^{\frac{n}{2}-s}(-m_{1}\frac{p^s-1}{t}-1)-1,  \ \text{if} \ F^{*}(-u)\neq 0 \\
  & \ \ \ \ \ \ \ \ \ \ \ \ \ \ \ \ \ \ \ \ \ \ \ \ \ \ \ \ \ \ \ \ \ \ \ \ \ \ \text{and} \ F^{*}(-u) \notin \cup_{j=1}^{m_{1}}w^{-z_{j}}H_{t}.\\
  \end{split}
  \right.
  \end{split}
\end{equation}
By Corollary 1, Lemma 1 and (18), (30), (31), any union of sets from $\{D_{0}^{F}$ $, D_{\beta H_{t}}^{F}, \beta \in \mathbb{F}_{p^s}^{*}\}$, denoted by $A$, is a $(p^n, k, \lambda, \mu)$ partial difference set, where
\begin{equation*}
  \begin{split}
      & k=(m_{1}\frac{p^s-1}{t}+m_{0})p^{n-s}+\varepsilon p^{\frac{n}{2}-s}(m_{0}p^s-(m_{1}\frac{p^s-1}{t}+m_{0}))-m_{0},\\
      & \lambda=(m_{1}\frac{p^s-1}{t}+m_{0})^{2}p^{n-2s}+\varepsilon p^{\frac{n}{2}-s}(p^s+(2m_{0}-3)m_{1}\frac{p^s-1}{t}-m_{0})-2m_{0},\\
      & \mu=(m_{1}\frac{p^s-1}{t}+m_{0})^{2}p^{n-2s}+\varepsilon p^{\frac{n}{2}-s}((2m_{0}-1)m_{1}\frac{p^s-1}{t}+m_{0}),
  \end{split}
\end{equation*}
$0\leq m_{1}\leq t$ is the number of distinct $D_{\beta H_{t}}^{F}$ in $A$, $m_{0}=1$ if $ D_{0}^{F} \subseteq A$ and $m_{0}=0$ if $D_{0}^{F} \not \subseteq A$.
\qed
\end{proof}

\begin{remark}\label{Remark7}
Let $m$ be a positive even integer, $s=2jr$ for some positive integers $j, r$, $t\geq 2$ and $t \mid (p^{j}+1)$, where $j$ is the smallest such positive integer. By the analysis on nonsingular quadratic forms in Section 2, we can see that any nonsingular diagonal quadratic form $G: \mathbb{F}_{p^s}^{m}\rightarrow \mathbb{F}_{p^s}$ defined as $G(x_{1}, \dots, x_{m})=a_{1}x_{1}^{2}+\dots+a_{m}x_{m}^{2}$ satisfies the corresponding condition of Theorem 4, where $a_{i} \in \mathbb{F}_{p^s}^{*}, 1\leq i \leq m$. Then by Theorem 4, the preimage set of any coset of $H_{t}$ for $G$ is a partial difference set. As we illustrate in Section 2, any nonsingular quadratic form $F: \mathbb{F}_{p^s}^{m}\rightarrow \mathbb{F}_{p^s}$ is linear equivalent to some nonsingular diagonal quadratic form $G$, thus the preimage set of any coset of $H_{t}$ for $F$ is a partial difference set, which is the result given in Proposition 2 (v). Furthermore, by Theorem 4, any union of sets from $\{D_{0}^{F}, D_{\beta H_{t}}^{F}, \beta \in \mathbb{F}_{p^s}^{*}\}$ is a partial difference set.
\end{remark}

When $s=2jr$ for some positive integers $j, r$, $t\geq 2$ and $t \mid (p^{j}+1)$, where $j$ is the smallest such positive integer, it is easy to see that the (non)-quadratic vectorial dual-bent functions defined by (3), and the (non)-quadratic vectorial dual-bent functions defined by (8) for which $2s \mid n$ and $\eta_{n}(\alpha_{i}), 1\leq i \leq 3$ are the same, all satisfy the corresponding condition of Theorem 4, thus we obtain the following explicit constructions of partial difference sets.

\begin{corollary}\label{Corollary5}
Let $p$ be an odd prime, $m, s, t$ be positive integers which satisfy that $s$ is a divisor of $m$, $s=2jr$ for some positive integers $j, r$, $t\geq 2$ and $t \mid (p^{j}+1)$, where $j$ is the smallest such positive integer. Let $H_{t}=\{x^{t}: x \in \mathbb{F}_{p^s}^{*}\}$. Define $F: \mathbb{F}_{p^m} \times \mathbb{F}_{p^m}\rightarrow \mathbb{F}_{p^s}$ as
 $F(x, y)=Tr_{s}^{m}(axy^{e})$,
where $a \in \mathbb{F}_{p^m}^{*}, gcd(e, p^m-1)=1$. Then for any $\beta \in \mathbb{F}_{p^s}^{*}$, $D_{\beta H_{t}}^{F}=\{(x, y) \in \mathbb{F}_{p^m} \times \mathbb{F}_{p^m}: F(x, y) \in \beta H_{t}\}$ is a $(p^{2m}, k, \lambda, \mu)$ partial difference set, where
\begin{equation*}
  \begin{split}
      & k=\frac{p^s-1}{t}p^{2m-s}-p^{m-s}\frac{p^s-1}{t},\\
      & \lambda=(\frac{p^s-1}{t})^{2}p^{2m-2s}+p^{m-s}(p^s-3\frac{p^s-1}{t}),\\
      & \mu=(\frac{p^s-1}{t})^{2}p^{2m-2s}-p^{m-s}\frac{p^s-1}{t}.
  \end{split}
\end{equation*}
Furthermore, any union of sets from $\{D_{0}^{F}, D_{\beta H_{t}}^{F}, \beta \in \mathbb{F}_{p^s}^{*}\}$, denoted by $A$, is a $(p^{2m}, k, \lambda, \mu)$ partial difference set, where
\begin{equation*}
  \begin{split}
      & k=(m_{1}\frac{p^s-1}{t}+m_{0})p^{2m-s}+p^{m-s}(m_{0}p^s-(m_{1}\frac{p^s-1}{t}+m_{0}))-m_{0},\\
      & \lambda=(m_{1}\frac{p^s-1}{t}+m_{0})^{2}p^{2m-2s}+p^{m-s}(p^s+(2m_{0}-3)m_{1}\frac{p^s-1}{t}-m_{0})-2m_{0},\\
      & \mu=(m_{1}\frac{p^s-1}{t}+m_{0})^{2}p^{2m-2s}+p^{m-s}((2m_{0}-1)m_{1}\frac{p^s-1}{t}+m_{0}),
  \end{split}
\end{equation*}
$0\leq m_{1}\leq t$ is the number of distinct $D_{\beta H_{t}}^{F}$ in $A$, $m_{0}=1$ if $ D_{0}^{F} \subseteq A$ and $m_{0}=0$ if $D_{0}^{F} \not \subseteq A$.
\end{corollary}

\begin{corollary}\label{Corollary6}
Let $p$ be an odd prime, $m, n, s, t$ be positive integers which satisfy that $s \mid m$, $2s \mid n$, $s=2jr$ for some positive integers $j, r$, $t\geq 2$ and $t \mid (p^{j}+1)$, where $j$ is the smallest such positive integer. Let $H_{t}=\{x^{t}: x \in \mathbb{F}_{p^s}^{*}\}$. Define $F_{i} \ (i \in \mathbb{F}_{p^s}): \mathbb{F}_{p^n}\rightarrow \mathbb{F}_{p^s}$ as $F_{0}(x)=Tr_{s}^{n}(\alpha_{1}x^{2})$, $F_{i}(x)=Tr_{s}^{n}(\alpha_{2}x^{2})$ if $i$ is a square in $\mathbb{F}_{p^s}^{*}$, $F_{i}(x)=Tr_{s}^{n}(\alpha_{3}x^{2})$ if $i$ is a non-square in $\mathbb{F}_{p^s}^{*}$, where $\alpha_{j} \in \mathbb{F}_{p^n}^{*}, 1\leq j\leq 3$ are all squares or non-squares. Define $G: \mathbb{F}_{p^m} \times \mathbb{F}_{p^m}\rightarrow \mathbb{F}_{p^s}$ as $G(y_{1}, y_{2})=Tr_{s}^{m}(\beta y_{1}L(y_{2}))$, where $\beta \in \mathbb{F}_{p^m}^{*}$ and $L(x)=\sum a_{i}x^{q^{i}} \ (q=p^s)$ is a $q$-polynomial over $\mathbb{F}_{p^m}$ inducing a permutation of $\mathbb{F}_{p^m}$. Let $T: \mathbb{F}_{p^n} \times \mathbb{F}_{p^m} \times \mathbb{F}_{p^m}\rightarrow \mathbb{F}_{p^s}$ be defined as $T(x, y_{1}, y_{2})=F_{Tr_{s}^{m}(\gamma y_{2}^{2})}(x)+G(y_{1}, y_{2})$, where $\gamma \in \mathbb{F}_{p^s}^{*}$, and let $\varepsilon=(-1)^{n-1}\epsilon^{n}\eta_{n}(\alpha_{1})$, where $\epsilon=1$ if $p\equiv 1 \ (mod \ 4)$ and $\epsilon=\sqrt{-1}$ if $p\equiv 3 \ (mod \ 4)$, $\eta_{n}$ is the multiplicative quadratic character of $\mathbb{F}_{p^n}$. Then for any $\beta \in \mathbb{F}_{p^s}^{*}$, $D_{\beta H_{t}}^{T}=\{(x, y_{1}, y_{2}) \in \mathbb{F}_{p^n} \times \mathbb{F}_{p^m} \times \mathbb{F}_{p^m}: T(x, y_{1}, y_{2}) \in \beta H_{t}\}$ is a $(p^{n+2m}, k, \lambda, \mu)$ partial difference set, where
\begin{equation*}
  \begin{split}
      & k=\frac{p^s-1}{t}p^{n+2m-s}-\varepsilon p^{\frac{n}{2}+m-s}\frac{p^s-1}{t},\\
      & \lambda=(\frac{p^s-1}{t})^{2}p^{n+2m-2s}+\varepsilon p^{\frac{n}{2}+m-s}(p^s-3\frac{p^s-1}{t}),\\
      & \mu=(\frac{p^s-1}{t})^{2}p^{n+2m-2s}-\varepsilon p^{\frac{n}{2}+m-s}\frac{p^s-1}{t}.
  \end{split}
\end{equation*}
Furthermore, any union of sets from $\{D_{0}^{T}, D_{\beta H_{t}}^{T}, \beta \in \mathbb{F}_{p^s}^{*}\}$, denoted by $A$, is a $(p^{n+2m}, k, \lambda, \mu)$ partial difference set, where
\begin{equation*}
  \begin{split}
      & k=(m_{1}\frac{p^s-1}{t}+m_{0})p^{n+2m-s}+\varepsilon p^{\frac{n}{2}+m-s}(m_{0}p^s-(m_{1}\frac{p^s-1}{t}+m_{0}))-m_{0},\\
      & \lambda=(m_{1}\frac{p^s-1}{t}+m_{0})^{2}p^{n+2m-2s}+\varepsilon p^{\frac{n}{2}+m-s}(p^s+(2m_{0}-3)m_{1}\frac{p^s-1}{t}-m_{0})\\
      & \ \ \ \ \ -2m_{0},\\
      & \mu=(m_{1}\frac{p^s-1}{t}+m_{0})^{2}p^{n+2m-2s}+\varepsilon p^{\frac{n}{2}+m-s}((2m_{0}-1)m_{1}\frac{p^s-1}{t}+m_{0}),
  \end{split}
\end{equation*}
$0\leq m_{1}\leq t$ is the number of distinct $D_{\beta H_{t}}^{T}$ in $A$, $m_{0}=1$ if $ D_{0}^{T} \subseteq A$ and $m_{0}=0$ if $D_{0}^{T} \not \subseteq A$.
\end{corollary}

We give two examples by using Corollaries 5 and 6.

\begin{example}\label{Example3}
Let $p=3, m=8, s=4, t=5$, then $p, m, s, t$ satisfy the condition in Corollary 5. Let $F: \mathbb{F}_{3^8} \times \mathbb{F}_{3^8}\rightarrow \mathbb{F}_{3^4}$ be defined by $F(x, y)=Tr_{4}^{8}(xy^{7})$. By Corollary 5, for any $\beta \in \mathbb{F}_{3^4}^{*}$, $D_{\beta H_{5}}^{F}$ is a $(43046721, 8501760, 1682289, 1678320)$ partial difference set and $D_{0}^{F}\cup D_{wH_{5}}^{F}\cup D_{w^{2}H_{5}}^{F}$ is a $(43046721, 17541440, 7148815,$ $7147602)$ partial difference set, where $w$ is a primitive element of $\mathbb{F}_{3^4}$.
\end{example}

\begin{example}\label{Example4}
Let $p=5, m=4, n=8, s=2, t=3$, then $p, m, n, s, t$ satisfy the condition in Corollary 6. Let $F_{i} \ (i \in \mathbb{F}_{5^2}): \mathbb{F}_{5^8} \rightarrow \mathbb{F}_{5^2}$ be defined by $F_{0}(x)=Tr_{2}^{8}(x^{2})$, $F_{i}(x)=Tr_{2}^{8}(w^{2}x^{2})$ if $i \in \mathbb{F}_{5^2}^{*}$, where $w$ is a primitive element of $\mathbb{F}_{5^8}$. Define $T: \mathbb{F}_{5^8} \times \mathbb{F}_{5^4} \times \mathbb{F}_{5^4}\rightarrow \mathbb{F}_{5^2}$ as $T(x, y_{1}, y_{2})=F_{Tr_{2}^{4}(y_{2}^{2})}(x)+Tr_{2}^{4}(y_{1}y_{2})=(Tr_{2}^{4}(y_{2}^{2}))^{24}Tr_{2}^{8}((w^{2}-1)x^{2})+Tr_{2}^{8}(x^{2})+Tr_{2}^{4}(y_{1}y_{2})$. By Corollary 6, for any $\beta \in \mathbb{F}_{5^2}^{*}$, $D_{\beta H_{3}}^{T}$ is a $(152587890625, 48828250000, 15624984375, 15625125000)$ partial difference set and $D_{wH_{3}}^{T}\cup D_{w^{2}H_{3}}^{T}$ is a $(152587890625, 97656500000, $ $62500359375, 62500250000)$ partial difference set.
\end{example}

In \cite{Cesmelioglu1}, \c{C}e\c{s}melio\u{g}lu \emph{et al.} showed that for a weakly regular $p$-ary bent function $f$ with $f(0)=0$, if $f$ is an $l$-form with $gcd(l-1, p-1)=1$ for some integer $1\leq l\leq p-1$, then $f$ (seen as a vectorial bent function from $V_{n}$ to $V_{1}$) is vectorial dual-bent. In the rest of this section, we prove that the converse also holds, which answers one open problem proposed in \cite{Cesmelioglu2}.

\begin{theorem}\label{Theorem5}
Let $f: V_{n}\rightarrow \mathbb{F}_{p}$ be a weakly regular vectorial dual-bent function with $f(0)=0$. Then $f$ is an $l$-form with $gcd(l-1, p-1)=1$ for some integer $1\leq l\leq p-1$.
\end{theorem}
\begin{proof}
Since $f$ is a weakly regular bent function, there exists a constant $\nu_{f} \in \{\pm 1\}$ such that $W_{f}(a)=\nu_{f} \sqrt{p^{*}}^{n}\zeta_{p}^{f^{*}(a)}$ for any $a \in V_{n}$, where $p^{*}=\eta_{1}(-1)p$, $\eta_{1}$ denotes the multiplicative quadratic character of $\mathbb{F}_{p}$. For any $\beta \in \mathbb{F}_{p}^{*}$ and $a \in V_{n}$, we have
\begin{equation}\label{32}
\begin{split}
 \sum_{x \in V_{n}}\zeta_{p}^{\beta f(x)-\langle a, x\rangle_{n}}
  & =\sum_{x \in V_{n}}\zeta_{p}^{\beta (f(x)-\langle \beta^{-1}a, x\rangle_{n})}\\
  & =\phi_{\beta}(W_{f}(\beta^{-1}a))\\
  & =\phi_{\beta}(\nu_{f}\sqrt{p^{*}}^{n}\zeta_{p}^{f^{*}(\beta^{-1}a)})\\
  & =\nu_{f}\sqrt{p^{*}}^{n}\eta_{1}(\beta)^{n}\zeta_{p}^{\beta f^{*}(\beta^{-1}a)},
  \end{split}
\end{equation}
where $\phi_{\beta}$ is the automorphism of the cyclotomic field $\mathbb{Q}(\zeta_{p})$ defined by $\phi_{\beta}(\zeta_{p})=\zeta_{p}^{\beta}, \phi_{\beta}(u)=u, u \in \mathbb{Q}$, and in the last equation we use the well-known fact that $\phi_{\beta}(\sqrt{p^{*}})=\eta_{1}(\beta)\sqrt{p^{*}}$ (see \cite{Ireland}).
The above equation shows that $\beta f$ is also weakly regular bent. Furthermore, since $f$ is vectorial dual-bent, we have
\begin{equation}\label{33}
  \begin{split}
  \sum_{x \in V_{n}}\zeta_{p}^{\beta f(x)-\langle a, x\rangle_{n}}=\nu_{\beta f}\sqrt{p^{*}}^{n}\zeta_{p}^{(\beta f)^{*}(a)}=\nu_{\beta f}\sqrt{p^{*}}^{n}\zeta_{p}^{\sigma(\beta)f^{*}(a)},
  \end{split}
\end{equation}
where $\nu_{\beta f}\in \{\pm 1\}$ is a constant and $\sigma$ is some permutation over $\mathbb{F}_{p}^{*}$.

By (32) and (33), we have $\nu_{\beta f}=\nu_{f}\eta_{1}(\beta)^{n}$ and
\begin{equation}\label{34}
  f^{*}(x)=\beta (\sigma(\beta))^{-1}f^{*}(\beta^{-1}x) \ \text{for any} \ \beta \in \mathbb{F}_{p}^{*}.
\end{equation}

When $\sigma$ is an identity map, by (34), we can see that $f^{*}$ is a $(p-1)$-form. Note that $gcd(p-2, p-1)=1$. By Proposition 5 and its proof of \cite{Tang}, we have that $(f^{*})^{*}(x)=f(-x)$ is also a $(p-1)$-form. Since $f(-x)$ is a $(p-1)$-form, it is easy to see that $f(x)=f(-x)$. Therefore, $f$ is an $l$-form with $gcd(l-1, p-1)=1$, where $l=p-1$.

In the following we consider the case that $\sigma$ is not an identity map. Let $\{\alpha_{1}, \dots, \alpha_{n}\}$ be a basis of $V_{n}$ over $\mathbb{F}_{p}$ and define $g(x_{1}, \dots, x_{n})=f^{*}(x_{1}\alpha_{1}+\dots+x_{n}\alpha_{n})$. By (34), we have
\begin{equation}\label{35}
  g(x_{1}, \dots, x_{n})=\beta (\sigma(\beta))^{-1}g(\beta^{-1}x_{1}, \dots, \beta^{-1}x_{n}) \ \ \text{for any} \ \beta \in \mathbb{F}_{p}^{*}.
\end{equation}
Assume that the algebraic normal form of $g$ is
\begin{equation}\label{36}
\begin{split}
  g(x_{1}, \dots, x_{n})&=\sum_{(j_{1}, \dots, j_{n})\in \mathbb{F}_{p}^{n}}a_{(j_{1}, \dots, j_{n})}\prod_{i=1}^{n}x_{i}^{j_{i}}\\
  &=a_{(0, \dots, 0)}+\sum_{d=1}^{n(p-1)}\sum_{(j_{1}, \dots, j_{n})\in \mathbb{F}_{p}^{n}: j_{1}+\dots+j_{n}=d}a_{(j_{1}, \dots, j_{n})}\prod_{i=1}^{n}x_{i}^{j_{i}},
\end{split}
\end{equation}
where $a_{(j_{1}, \dots, j_{n})}\in \mathbb{F}_{p}$.
Then
\begin{equation}\label{37}
\begin{split}
  &\beta (\sigma(\beta))^{-1}g(\beta^{-1}x_{1}, \dots, \beta^{-1}x_{n})=\beta (\sigma(\beta))^{-1}a_{(0, \dots, 0)}\\
  &\ +\sum_{d=1}^{n(p-1)}\sum_{(j_{1}, \dots, j_{n})\in \mathbb{F}_{p}^{n}: j_{1}+\dots+j_{n}=d}a_{(j_{1}, \dots, j_{n})}\beta^{1-d}(\sigma(\beta))^{-1}\prod_{i=1}^{n}x_{i}^{j_{i}}.
  \end{split}
\end{equation}
By (35)-(37), we have that
$a_{(0, \dots, 0)}=0$ and if for $(j_{1}, \dots, j_{n})$ with $j_{1}+\dots+j_{n}=d$ we have $a_{(j_{1}, \dots, j_{n})}\neq 0$, then $\sigma(\beta)=\beta^{1-d}$ for any $\beta \in \mathbb{F}_{p}^{*}$, which implies that there exists an integer $1<d_{0}<p-1$ with $gcd(d_{0}-1, p-1)=1$ such that $\sigma(\beta)=\beta^{1-d_{0}}, \beta \in \mathbb{F}_{p}^{*}$. Then by (34), we can see that $f^{*}$ is a $d_{0}$-form. By Proposition 5 and its proof of \cite{Tang}, we have that $(f^{*})^{*}(x)=f(-x)$ is an $l$-form, where $(d_{0}-1)(l-1)\equiv 1 \ mod \ (p-1)$. Since $f(-x)$ is an $l$-form with $gcd(l-1, p-1)=1$, it is easy to see that $f(x)=f(-x)$. Therefore, $f$ is an $l$-form with $gcd(l-1, p-1)=1$.
\qed
\end{proof}

\section{Conclusion}
\label{sec:5}
In this paper, we further studied vectorial dual-bent functions and partial difference sets.

(1) We provided a new explicit construction of vectorial dual-bent functions (Theorem 1).

(2) For vectorial dual-bent functions $F: V_{n}\rightarrow \mathbb{F}_{p^s}$ with certain additional properties, we proved that if the permutation $\sigma$ of $\mathbb{F}_{p^s}^{*}$ given by $(F_{c})^{*}=(F^{*})_{\sigma(c)}, c \in \mathbb{F}_{p^s}^{*}$ (where $F^{*}$ is a vectorial dual of $F$) is an identity map, then the preimage set of any subset of $\mathbb{F}_{p^s}$ for $F$ forms a partial difference set (Theorem 2).

(3) For vectorial dual-bent functions $F: V_{n}\rightarrow \mathbb{F}_{p^s}$ with certain additional properties, if the permutation $\sigma$ of $\mathbb{F}_{p^s}^{*}$ given by $(F_{c})^{*}=(F^{*})_{\sigma(c)}, c \in \mathbb{F}_{p^s}^{*}$ (where $F^{*}$ is a vectorial dual of $F$) is not an identity map but it satisfies some conditions, we proved that the preimage set of the squares (non-squares) in $\mathbb{F}_{p^s}^{*}$ for $F$ and the preimage set of any coset of some subgroup of $\mathbb{F}_{p^s}^{*}$ for $F$ form partial difference sets (Theorems 3, 4 and Corollary 2). Furthermore, explicit constructions of partial difference sets are yielded from some (non)-quadratic vectorial dual-bent functions (Corollaries 3, 4, 5, 6).

(4) We proved that if $f$ is a $p$-ary weakly regular vectorial dual-bent function with $f(0)=0$, then $f$ is an $l$-form with $gcd(l-1, p-1)=1$ for some integer $1\leq l\leq p-1$ (Theorem 5), which answers one open problem proposed in \cite{Cesmelioglu2}.

In particular, in Remarks 1, 3, 4, we illustrated that almost all the results of using weakly regular $p$-ary bent functions to construct partial difference sets are special cases of our results.

\begin{acknowledgements}
This research is supported by the National Key Research and Development Program of China (Grant No. 2018YFA0704703), the National Natural Science Foundation of China (Grant Nos. 12141108 and 61971243), the Natural Science Foundation of Tianjin (20JCZDJC00610), the Fundamental Research Funds for the Central Universities of China (Nankai University), and the Nankai Zhide Foundation.
\end{acknowledgements}




\end{document}